%% file: ms.tex
\pdfoutput=1
\documentclass[10pt,final]{article}
\usepackage[T1]{fontenc}
\usepackage[utf8]{inputenc}
\usepackage{etoolbox}
\usepackage[margin=1in]{geometry}
\usepackage{graphicx}
\usepackage{xcolor}
\usepackage{inconsolata}
\usepackage{lmodern}
\usepackage{mathtools}
\usepackage{amssymb}
\usepackage{amsthm}
\usepackage{amsfonts}
\usepackage{bm}
\usepackage[linesnumbered,ruled]{algorithm2e}
\usepackage{hyperref}
\usepackage{url}
\usepackage[capitalize]{cleveref}
\usepackage[bf]{caption}
\usepackage{booktabs}
\usepackage{titling}
\usepackage[style=alphabetic,backref=true,maxnames=99,maxalphanames=5,natbib=true]{biblatex}

\DefineBibliographyStrings{english}{
backrefpage = {p.},
backrefpages = {pp.}
}
\definecolor{citegreen}{HTML}{208054}
\definecolor{citeblue}{HTML}{0055cc}

\input{macros}

\newcommand{\email}[1]{\href{mailto:#1}{\color{black}{\texttt{#1}}}}

\title{%
Learning the structure of any Hamiltonian \\from minimal assumptions
}

\author{%
Andrew Zhao\\
\email{azhao@sandia.gov}\\
Sandia National Laboratories
}

\date{April 21, 2025}

\hypersetup{
pdfauthor={Andrew Zhao},
colorlinks=true,
citecolor=citegreen,
linkcolor=citeblue,
linktocpage=true
}

\begin{document}

\maketitle

\begin{abstract}
    We study the problem of learning an unknown quantum many-body Hamiltonian $H$ from black-box queries to its time evolution $e^{-\mathrm{i} H t}$. Prior proposals for solving this task either impose some assumptions on $H$, such as its interaction structure or locality, or otherwise use an exponential amount of computational postprocessing. In this paper, we present algorithms to learn any $n$-qubit Hamiltonian, which do not need to know the Hamiltonian terms in advance, nor are they restricted to local interactions. Our algorithms are efficient as long as the number of terms $m$ is polynomially bounded in the system size $n$. We consider two models of control over the time evolution:~the first has access to time reversal ($t < 0$), enabling an algorithm that outputs an $\epsilon$-accurate classical description of $H$ after querying its dynamics for a total of $\widetilde{\mathcal{O}}(m/\epsilon)$ evolution time. The second access model is more conventional, allowing only forward-time evolutions;~our algorithm requires $\widetilde{\mathcal{O}}(\|H\|^3/\epsilon^4)$ evolution time in this setting. Central to our results is the recently introduced concept of a pseudo-Choi state of $H$. We extend the utility of this learning resource by showing how to use it to learn the Fourier spectrum of $H$, how to achieve nearly Heisenberg-limited scaling with it, and how to prepare it even under our more restricted access models.
\end{abstract}

\newpage

\tableofcontents

\newpage

\input{introduction}

\input{background}

\input{psuedo-choi_method}

\input{structure_learning}

\input{residual_hamiltonian}

\input{time-reversal-free_algorithm}

\input{controlization}

\input{acknowledgments}

\printbibliography
\addcontentsline{toc}{section}{References}

\appendix

\input{appendix_chernoff}

\input{galactic_algorithm}

\end{document}

%% file: macros.tex
\theoremstyle{plain}
\newtheorem{theorem}{Theorem}[section]
\newtheorem{lemma}[theorem]{Lemma}
\newtheorem{corollary}[theorem]{Corollary}
\newtheorem{proposition}[theorem]{Proposition}

\theoremstyle{definition}
\newtheorem{definition}[theorem]{Definition}
\newtheorem{problem}[theorem]{Problem}

\newtheorem{remark}[theorem]{Remark}

\renewcommand{\O}{\mathcal{O}}
\newcommand{\Ot}{\widetilde{\mathcal{O}}}
\newcommand{\Omegat}{\widetilde{\Omega}}
\newcommand{\Thetat}{\widetilde{\Theta}}
\DeclareMathOperator{\poly}{poly}
\DeclareMathOperator{\polylog}{polylog}
\newcommand{\R}{\mathbb{R}}
\newcommand{\C}{\mathbb{C}}

\newcommand{\I}{\mathbb{I}}

\renewcommand{\Re}{\operatorname{Re}}

\renewcommand{\i}{\mathrm{i}}

\newcommand{\U}{\mathrm{U}}

\newcommand{\Cl}{\mathrm{Cl}}

\DeclareMathOperator{\supp}{supp}

\DeclareMathOperator*{\E}{\mathbb{E}}
\DeclareMathOperator{\V}{Var}

\DeclareMathOperator{\tr}{tr}

\newcommand{\op}[2]{\ket{#1}\!\bra{#2}}
\newcommand{\ip}[2]{\langle #1 | #2 \rangle}
\newcommand{\melem}[3]{\langle #1 | #2 | #3 \rangle}
\newcommand{\ev}[2]{\melem{#2}{#1}{#2}}
\newcommand{\ket}[1]{| #1 \rangle}
\newcommand{\bra}[1]{\langle #1 |}

\newcommand{\comp}{\mathsf{C}}
\let\oldl\l
\renewcommand{\l}[1]{\mathopen{}\left#1}
\renewcommand{\r}[1]{\right#1\mathclose{}}

\newcommand{\prep}{\mathsf{PREP}}
\newcommand{\sel}{\mathsf{SEL}}

\DeclareMathOperator{\alg}{\mathcal{A}}
\newcommand{\Pauli}[1]{\mathcal{P}_{#1}}
\newcommand{\pchoi}{\Psi}
\newcommand{\pchoiref}{\pchoi_{\mathrm{ref}}}
\newcommand{\norm}{\mathcal{N}}
\newcommand{\ctrl}[2]{\texttt{ctrl}_{#1}(#2)}
\newcommand{\ttotal}{t_{\mathrm{total}}}
\newcommand{\tmin}{t_{\mathrm{min}}}
\newcommand{\Nexp}{N_{\mathrm{exp}}}
\newcommand{\nanc}{n_{\mathrm{anc}}}

%% file: introduction.tex
\section{Introduction}\label{sec:intro}

Identifying an unknown Hamiltonian is a central task in quantum physics. As the Hamiltonian fully characterizes the interactions within any closed system, learning it provides a tractable description of a potentially highly complex quantum system. This is in contrast to the more general task of quantum process tomography, which learns the dynamical map itself but is exponentially expensive in the size of the system, even when the system is closed~\cite{baldwin2014quantum,bavaresco2022unitary,haah2023query}. Applications of Hamiltonian learning include quantum metrology and sensing~\cite{giovannetti2011advances,degen2017quantum}, quantum device engineering~\cite{shulman2014suppressing,innocenti2020supervised}, and certifying quantum computations~\cite{wiebe2014hamiltonian}. By developing more accurate and efficient methods for understanding quantum systems, we can enable the design and control of such systems with finer precision.

Historically, Hamiltonian learning draws from the subject of quantum metrology~\cite{ramsey1950molecular,caves1981quantum,bollinger1996optimal}, extending to heuristic protocols for learning many-body models with familiar types of interactions~\cite{granade2012robust,wiebe2014quantum,holzapfel2015scalable,wang2017experimental,bairey2019learning}. Recently, there has been a surge of interest in learning Hamiltonians of large many-body systems with more complicated interactions and rigorous theoretical guarantees~\cite{anshu2021sample,haah2024learning,gu2024practical,huang2023learning,caro2024learning,dutkiewicz2023advantage,bakshi2024learning,bakshi2024structure}. This body of work substantially broadens the applicability of this task, especially as experimental capabilities in engineering large-scale quantum devices continues to improve. From a theoretical perspective, such works explore fundamental questions in the learnability of quantum systems, such as:~what classes of Hamiltonians can be learned? How efficiently and accurately can this task be carried out? And what level of control or access to the system suffices? This paper aims to shed some light on these foundational questions.

Broadly speaking, an unknown Hamiltonian is accessed either through its static or dynamical behavior. In the former, we are tasked with deducing the Hamiltonian when the system is at equilibrium~\cite{alhambra2023quantum}. This is typically given through its Gibbs state at some temperature~\cite{anshu2021sample,haah2024learning,gu2024practical,bakshi2024learning}, although learning from eigenstates has also been considered~\cite{qi2019determining}. In the latter model, we are given black-box access to the time evolution operator $e^{-\i H t}$ of the unknown Hamiltonian $H$, for some controllable amount of time $t$. Again, we are tasked to output a classical description of $H$, this time given the ability to prepare some initial probe states, evolving them under the dynamics of $H$, and measuring in some bases. Both versions of Hamiltonian learning have different merits and challenges, and are appropriate depending on the experimental scenario at hand. For example, the static model is more experimentally friendly, arising in many natural settings where one has limited control of the system. On the other hand, dynamical access requires finer experimental control but can serve as a more powerful learning resource~\cite{haah2024learning,huang2023learning}. Learning from the time evolution operator also has close connections to applications in quantum simulation and algorithms. For the remainder of this paper, we will focus exclusively on the dynamical paradigm.

In this context, the primary cost of a Hamiltonian learning algorithm is the total evolution time, $t_\mathrm{total}$. This is the total amount of time under which we evolve by $e^{-\i H t}$. We can understand this cost in relation to the fractional- or continuous-query model~\cite{berry2014exponential}, wherein queries to the time evolution operator cost less for smaller $t$. Note that this is part of, but not the same as, the total time complexity of the algorithm, which includes the cost of additional quantum and classical computation (before, during, and after querying the system). In this work, we will not closely inspect the details of the computational complexity;~we will only be concerned whether an algorithm is computationally efficient (polynomially scaling in all parameters) or not. Other figures of merit include the time resolution $\tmin$, which dictates how finely an experimenter must be able to control the time evolution;~$\Nexp$, the total number of prepare--evolve--measure experiments to run;~and $\nanc$, the number of ancilla qubits used either to perform supporting quantum computation/control or to coherently store queries in quantum memory.

Most Hamiltonian learning results operate under the assumption that the structure of the unknown Hamiltonian is fully or at least partially known. The structure of a Hamiltonian $H = \sum_{a=1}^m \lambda_a E_a$ is the set of $m$ terms $E_a$, typically chosen from some operator basis, for which the coefficients $\lambda_a$ are nonzero. For some algorithms, the learning guarantees can be lost in the worst case when $H$ has even a single unidentified term~\cite{haah2024learning,huang2023learning}. For others, the algorithm may succeed at learning the parameters of a partially known structure, but there is no guarantee for efficiently learning the unidentified terms~\cite{dasilva2011practical,bairey2019learning,zubida2021optimal}. To address this deficiency, Bakshi, Liu, Moitra, and Tang~\cite{bakshi2024structure} formalized the problem of Hamiltonian \emph{structure} learning, and they provided an efficient algorithm to solve it. Specifically, when $H$ is a constant-local $n$-qubit Hamiltonian (with no restriction on the range of its local interactions) they showed that it is possible to identify all local terms $E_a$ with $|\lambda_a| > \epsilon$, and furthermore estimate the values of all such $\lambda_a$ to additive error $\epsilon$. In terms of evolution time, their algorithm costs $\ttotal = \O(r \log(n) / \epsilon)$, where $r$ is an effective sparsity parameter of the qubits' interaction graph.

With nearly optimal scaling, this breakthrough result essentially completely solves the problem for local spin Hamiltonians. However, it remained an open question as to how one performs structure learning for even broader classes of Hamiltonians. For example, suppose that we assume $H$ is $k$-local, but there are in fact a few $k'$-local terms, where $k' > k$. This can occur, for instance, by (undesired) spin-squeezing effects~\cite{katz2022body}. In this setting, the algorithm of \cite{bakshi2024structure} can learn all the $k$-local terms, but not the $k'$-local terms. Furthermore, even if one were given the value of $k'$, if say $k' = \omega(1)$ then the complexity of their algorithm would be superpolynomial.

More generally, we pose the scenario where $H$ is completely arbitrary, with no locality restriction on its terms. Even if $m \leq \poly(n)$, it is unclear how to identify these terms from a pool of $2^{\O(n)}$ candidate terms within polynomial time. \cite{bakshi2024structure} refers to this problem as the ``sparse, nonlocal'' setting and suggests that their algorithm could be applied here as well;~however it is unclear how to avoid an exponential amount of computation using either their Goldreich--Levin-like algorithm~\cite{goldreich1989hardcore} or general shadow tomography ideas~\cite{aaronson2020shadow,buadescu2021improved}. To our knowledge, there are only a handful of proposals in the literature which can operate in this setting that at least avoid an exponential number of experiments~\cite{yu2023robust,caro2024learning,castaneda2023hamiltonian}. However, the algorithms of Caro~\cite{caro2024learning} and Castaneda and Wiebe~\cite{castaneda2023hamiltonian} would use an exponential amount of classical computation to solve this task, essentially brute-force checking over all possible terms. Meanwhile, the analysis by Yu, Sun, Han, and Yuan~\cite{yu2023robust} assumes that the $m$ terms are distributed uniformly among the set of Pauli operators, and therefore only guarantees that their algorithm succeeds with high probability (when $1/m$ is sufficiently small) over the draw of $H$ from this ensemble of Hamiltonians. Their numerics suggest good heuristic performance but does not rigorously guarantee the ability to structure-learn any given sparse, nonlocal Hamiltonian in the worst case.

We therefore ask:~can efficient structure learning be accomplished in this setting, and if so, what is the minimal amount of prior information required? Because we require an efficient algorithm, we only consider Hamiltonians for which $m \leq \poly(n)$, as otherwise merely writing down the description of $H$ is intractable. As an example of prior information, consider again the algorithm of \cite{bakshi2024learning}. Implicit in their problem input is a bound on the effective sparsity and a ``local one-norm'' of $H$ (i.e., the sum of strengths of all terms involving any given qubit). Bounds on these quantities are necessary in order to choose appropriate values for $\tmin$ and $\Nexp$. On the other hand, this aspect is trivial in non-structure-learning settings because knowing the structure is already more than sufficient to establish such bounds. We answer both questions by describing structure learning algorithms which essentially only require a bound on $m$ to run correctly. As long as $m \leq \poly(n)$, our algorithms are efficient in all metrics.

\subsection{Main results}

In this paper, we propose two algorithms for learning the structure of any unknown $n$-qubit Hamiltonian. Let us begin by providing a technical definition for what we mean by the structure of a Hamiltonian.

\begin{definition}
    An $n$-qubit \emph{Hamiltonian} is a Hermitian $2^n \times 2^n$ matrix
    \begin{equation}
        H = \sum_{a=1}^{4^n} \lambda_a E_a,
    \end{equation}
    where $E_a \in \Pauli{n} = \{\I, X, Y, Z\}^{\otimes n}$ are the unique \emph{Pauli terms} of $H$, and $\lambda_a \in \R$ are the \emph{parameters} of $H$. The labeling $a \in [4^n]$ is arbitrary;~defining the \emph{Pauli structure} of $H$ as the set $S = \{(E_a, \lambda_a) : E_a \neq \I \text{ and } \lambda_a \neq 0\}$, we take convention that the labels $a = 1, \ldots, |S|$ correspond to the elements of $S$. As a slight abuse of notation, we refer to $\lambda \in \R^{|S|}$ as the vector of nonzero parameters, which we can pad with zeros as necessary.
\end{definition}

Our definition of structure is understood with respect to an operator basis. In this paper we consider the Pauli basis $\Pauli{n}$, as it is arguably the most natural choice for studying many-qubit systems. Other choices of basis are also possible in principle, as long as they admit an efficient classical description;~such a choice should be thought of as part of the definition to the learning problem. In this paper, we will exclusively consider the Pauli structure, which is a quantum analogue to the Fourier spectrum of a classical Boolean function. We will be primarily concerned with \emph{sparse} ($|S| \leq \poly(n)$), \emph{nonlocal} ($\supp(E_a) \leq n$) Hamiltonians in this basis.

Note that our definition of $S$ explicitly excludes the identity component, as it merely constitutes an energy shift. This shift contributes to the time evolution $e^{-\i H t}$ as a global phase, and is therefore fundamentally unobservable within our black-box access models. However, let us point out that we will not be restricted to traceless Hamiltonians as inputs;~our techniques will essentially extract the traceless part of any given Hamiltonian ``for free.''

With these considerations in mind, let us define the learning problem as follows.

\begin{problem}\label{prob:1}
    Let $H = \sum_{a} \lambda_a E_a$ be an unknown $n$-qubit Hamiltonian such that $\|\lambda\|_\infty \leq 1$, $S$ its Pauli structure, and $\epsilon \in (0, 1)$ a desired accuracy parameter. In the \emph{Hamiltonian structure learning problem}, we are given as input:~query access to applications of $e^{-\i H t}$ for tunable $t$ and an integer $m \geq |S|$. We are tasked with outputting a description $\widehat{S} = \{(E_a, \widehat{\lambda}_a) : a = 1, 2, \ldots\}$, where no $E_a = \I$, that obeys the following properties with probability at least $2/3$:
    \begin{enumerate}
        \item if $|\lambda_a| > \epsilon$, then $(E_a, \widehat{\lambda}_a) \in \widehat{S}$;
        \item if $(E_a, \widehat{\lambda}_a) \in \widehat{S}$, then $|\lambda_a| > \epsilon/2$;
    \end{enumerate}
    with each $\widehat{\lambda}_a$ satisfying $|\widehat{\lambda}_a - \lambda_a| \leq \epsilon$. 
\end{problem}

The $2/3$ success probability is conventional;~by a standard median-boosting argument, any algorithm that solves this problem can be repeated $\O(\log(1/\delta))$ times to amplify the probability above $1 - \delta$~\cite{jerrum1986random}. For simplicity we have assumed that $m$ is a bound on the number of precisely nonzero terms in $H$;~the arguments in this paper will hold with minor modifications if we relax $m \to m_\epsilon$, the number of terms with $|\lambda_a| > \epsilon$, provided the additional promise that $\|\lambda\|_1 \leq \poly(n)$. Finally, we remark that the normalization guarantee $\|\lambda\|_\infty \leq 1$ is a promise on the largest Pauli coefficient in $H$, which essentially allows the experimentalist to choose the appropriate units of time and energy (c.f.~the expression $e^{-\i Ht}$).

The particular type of access to $e^{-\i H t}$ that we are given can significantly affect the learning guarantees for solving \cref{prob:1}. At a minimum, we will suppose discrete control over the Hamiltonian in the sense of \cite{dutkiewicz2023advantage}. This is the standard type of quantum control assumed in most settings.

\begin{definition}
    \emph{Discrete quantum control} is the ability to run the circuit
    \begin{equation}
        V_L (e^{-\i H t_L} \otimes \I) V_{L-1} \cdots V_1 (e^{-\i H t_1} \otimes \I) V_0,
    \end{equation}
    where $V_0, V_1, \ldots, V_L$ are discrete control operations (quantum gates) acting jointly on the $n$-qubit system and some $\nanc$ ancilla qubits, and the $t_j \geq 0$ are some sequence of tunable evolution times.
\end{definition}

Without some form of quantum control, \cite{dutkiewicz2023advantage} showed that any Hamiltonian learning algorithm cannot learn with precision beyond the standard quantum limit $\Omega(1/\epsilon^2)$.\footnote{Up to some technical considerations about $H$.} Along with discrete control, one of our algorithms uses an additional form of control to beat the standard limit:~the ability to implement time reversal.

\begin{definition}
    Access to \emph{time-reversal dynamics} is the ability to apply $e^{-\i H t}$ for both $t \geq 0$ and $t < 0$.
\end{definition}

For example, systems possessing a chiral symmetry $\Gamma$ obey $\Gamma H \Gamma^\dagger = -H$, which immediately gives rise to time-reversal dynamics. Alternatively, in the context of algorithmic certification, one typically has direct control over $t$ as a tunable, real parameter. The power of time reversal in learning was previously considered in \cite{schuster2023learning,cotler2023information,schuster2024random}, where exponential separations for certain unitary learning tasks were established between protocols with access to time reversal and those without. However, these results do not necessarily imply such a separation in our setting (i.e., learning Hamiltonians with sparse structures).

With access to discrete quantum control and time-reversal dynamics, we prove the following theorem.

\begin{theorem}[Structure learning with time reversal]\label{thm:tr_learning}
    There exists a quantum algorithm that solves \cref{prob:1} under the discrete control and time-reversal access model. The algorithm queries $e^{-\i H t}$ for a total evolution time of $\ttotal = \Ot(m/\epsilon)$, requires a time resolution of $|t| \geq \tmin = \Omegat(\epsilon/m^3)$, runs $\Nexp = \Ot(m^2 \log(1/\epsilon))$ experiments, and uses $\poly(n, m, 1/\epsilon)$ quantum and classical computation. The discrete control operations act on an ancillary computational register of $\nanc = n + \O(\log m)$ qubits.
\end{theorem}

This algorithm nearly achieves the Heisenberg limit of quantum metrology, up to a factor of $\log(1/\epsilon)$ hidden in the $\Ot(\cdot)$ notation.

In more restricted experimental settings, it is unrealistic to assume the ability to perform time reversal. Therefore, we consider a second access model that is more conventional: only having discrete quantum control and forward-time evolutions.

\begin{theorem}[Structure learning without time reversal]\label{thm:tf_learning}
    There exists a quantum algorithm that solves \cref{prob:1} under the discrete control access model alone. The algorithm queries $e^{-\i H t}$ for a total evolution time of $\ttotal = \Ot(\|H\|^3/\epsilon^4)$, requires a time resolution of $t \geq \tmin = \Omegat(\epsilon^4/\|H\|^5)$, runs $\Nexp = \Ot(\|H\|^4/\epsilon^4)$ experiments, and uses $\poly(n, m, 1/\epsilon)$ quantum and classical computation. The discrete control operations act on an ancillary computational register of $\nanc = n + \O(\log\log(\|H\|/\epsilon))$ qubits.
\end{theorem}

Above, we have phrased the complexity in terms of the operator norm $\|H\|$, mirroring related results in the literature~\cite{caro2024learning,castaneda2023hamiltonian}. Without knowledge of $\|H\|$, we may simply apply the bound $\|H\| \leq m$ instead (which holds due to our normalization convention $\|\lambda\|_\infty \leq 1$). We remark that it is not possible to make the reverse replacement in \cref{thm:tr_learning} with our current analysis.

The proof of \cref{thm:tf_learning} can in fact be strengthened to almost recover the conventional $1/\epsilon^2$ scaling in the precision. Although a strict asymptotic improvement in every metric, the resulting algorithm becomes ``galactic'':~for example, to get a scaling of $1/\epsilon^{2.01}$, constants on the order of $10^{60}$ appear.\footnote{On the other hand, smaller speedups feature more reasonable constants. For instance, achieving $\ttotal = \Ot(\|H\|^2/\epsilon^3)$ only incurs a factor of $4$ relative to the hidden constants of \cref{thm:tf_learning}.} We therefore treat this result as a curious observation, and we will defer any further discussion and analysis to \cref{sec:galactic_algorithm}.

\begin{corollary}[Galactic structure learning without time reversal]\label{cor:galactic_algorithm}
    There exists a quantum algorithm that solves \cref{prob:1} under the discrete control access model alone. For any real constant $p \geq 1$, the algorithm queries $e^{-\i H t}$ for a total evolution time of $\ttotal = \Ot((2^p \|H\|)^{1+2/p}/\epsilon^{2+2/p})$, requires a time resolution of $t \geq \tmin = \Omegat(\epsilon^4/(2^p \|H\|)^5)$, runs $\Nexp = \Ot((2^p \|H\|)^{2+2/p}/\epsilon^{2+2/p})$ experiments, and uses $\poly(n, m, 1/\epsilon)$ quantum and classical computation. The discrete control operations act on an ancillary computational register of $\nanc = n + \O(\log\log(\|H\|/\epsilon))$ qubits.
\end{corollary}

\subsection{Related work}

In this section we give an overview of prior works on Hamiltonian learning from real-time evolution, as well as some related topics. The literature on Hamiltonian learning is vast and so this discussion will not be comprehensive. \cref{tab:results} catalogs various prior results and shows how they compare to our algorithms. In the table, we say an algorithm can achieve structure learning (SL) depending on the assumed class of Hamiltonians. This is a subtle but important point:~any algorithm under the ``sparse, nonlocal'' category (regardless of whether it is labeled as SL-capable or not) can learn the structure of any constant-local Hamiltonian, simply by searching over the entire $n^{\O(1)}$-sized pool of candidate terms.\\

\begin{table}
    \centering
    \makebox[\linewidth][c]{%
    \begin{tabular}{l l l l c l l l}
        \toprule
        Reference & Hamiltonian class & SL? & Control? & $n_\mathrm{anc}$ & $\ttotal$ & $\tmin$ & $\Nexp$\\
        \hline
        \cite{zubida2021optimal}$^*$ & Geometrically local & Yes & None & 0 & $1/\epsilon^3$ & $\epsilon$ & $1/\epsilon^4$\\
        \cite{stilck2024efficient} & ------ & No & None & $0$ & $1/\epsilon^{2}$ & $\frac{1}{\polylog(1/\epsilon)}$ & $1/\epsilon^2$\\
        \cite{haah2024learning} & Low intersection & No & None & $0$ & $1/\epsilon^{2}$ & $1$ & $1/\epsilon^2$\\
        \cite{huang2023learning} & ------ & No & Discrete & $0$ & $1/\epsilon$ & $\sqrt{\epsilon}$ & $\polylog(1/\epsilon)$\\
        \cite{dutkiewicz2023advantage} & ------ & No & Continuous & $0$ & $1/\epsilon$ & $1$ & $\log(1/\epsilon)$\\
        \cite{bakshi2024structure} & \begin{tabular}{@{}l@{}}
            Effectively sparse, \\bounded strength
        \end{tabular} & Yes & Discrete & $0$ & $1/\epsilon$ & $1$ & $\log(1/\epsilon)$\\
        \cite{yu2023robust} & Sparse, nonlocal & Yes$^\dagger$ & Discrete & $0$ & $n\|H\|^3/\epsilon^{4}$ & $1/\|H\|$ & $n\|H\|^4/\epsilon^4$\\
        \cite{caro2024learning} & ------ & No$^\ddagger$ & Discrete & $n + \|H\|^2/\epsilon^2$ & $\|H\|^3/\epsilon^{4}$ & $1/\|H\|$ & $\|H\|^4/\epsilon^{4}$\\
        \cite{odake2024higher} & ------ & No & Discrete & $1$ & $m/\epsilon$ & $\epsilon^2/\|H\|^2$ & $m \log(1/\epsilon)$\\
        \cite{castaneda2023hamiltonian}$^\S$ & ------ & No$^\ddagger$ & $\ctrl{}{e^{\pm\i H t}}$ & $n$ & $\|H\|/\epsilon^{2}$ & $1/\|H\|$ & $\|H\|^2/\epsilon^{2}$\\
        \cite{castaneda2023hamiltonian}$^\S$ & ------ & No & $\ctrl{}{e^{\pm\i H t}}$ & $n + m$ & $\sqrt{m} /\epsilon$ & $1/\|H\|$ & $\|H\| \sqrt{m}/\epsilon$\\
        
        \hline\hline
        
        Here & Sparse, nonlocal & Yes & $\ctrl{}{e^{\pm \i H t}}$ & $n + \log m$ & $m/\epsilon$ & $1/m$ & $m^2 \log(1/\epsilon)$\\
        Here & ------ & Yes & Discrete, $\pm t$ & $n + \log m$ & $m/\epsilon$ & $\epsilon/m^3$ & $m^2 \log(1/\epsilon)$\\
        Here$^\P$ & ------ & Yes & Discrete & $n + \log\log(\|H\|/\epsilon)$ & $\|H\|^3/\epsilon^4$ & $\epsilon^4/\|H\|^5$ & $\|H\|^4/\epsilon^4$\\
        \bottomrule
    \end{tabular}
    }
    \caption{Nonexhaustive list of algorithms for learning an $n$-qubit Hamiltonian from time evolution. If the operator norm is not known, it can be replaced with $m \geq \|H\|$. Each Hamiltonian class in the table is subsumed by the one below it. Subdominant constants and polylogarithmic factors are suppressed.\\
    $^*$The structure-learning analysis was demonstrated posthoc in \cite{bakshi2024structure}.\\
    $^\dagger$Their analysis assumes that the $m$ unknown terms are distributed uniformly among all Pauli operators.\\
    $^\ddagger$SL can be achieved if given exponential classical resources (also incurs a factor of $n$ in $t_\mathrm{total}$ and $N_\mathrm{exp}$).\\
    $^\S$Their results were originally stated with $\ell_2$-error, which we have converted to $\ell_\infty$-error here.\\
    $^\P$Can be improved to $\ttotal = \|H\|^{1+o(1)}/\epsilon^{2+o(1)}$ and $\Nexp = \|H\|^{2+o(1)}/\epsilon^{2+o(1)}$, with potentially enormous constants.}
    \label{tab:results}
\end{table}

\textbf{Algorithms capable of (local) structure learning.} While \cite{bakshi2024structure} first put forth the explicit problem of Hamiltonian structure learning, they also recognized that prior results could be reanalyzed within this context. At the core of many of these results, including \cite{bakshi2024structure} itself, is the so-called time-derivative estimator~\cite{shabani2011estimation,dasilva2011practical,zubida2021optimal,stilck2024efficient}. This technique is based on the relation $\tr(P_a e^{-\i H t} \rho_a e^{\i H t}) \approx 2 \lambda_a t$ for small $t$, where $\rho_a = \frac{1}{2^n}(\I + \i E_a P_a)$ is a mixed state defined by any Pauli operator $P_a$ that anticommutes with the term $E_a$ to be learned. This relation does not depend on the locality of the terms, and therefore even a naive term-by-term approach can estimate any sparse, nonlocal Hamiltonian in polynomial time. However, achieving \emph{structure} learning by this approach requires that the size of the pool of candidate terms be polynomially bounded, which we explicitly do not assume in our setting. Furthermore, parallelizing the estimation of many $E_a$'s per experiment relies on a classical-shadows-like argument that only works for local terms~\cite{huang2020predicting}.

The algorithms of Caro~\cite{caro2024learning} and Castaneda and Wiebe~\cite{castaneda2023hamiltonian} experience similar advantages to, as well as limitations of, the aforementioned strategies. However, one particular benefit they enjoy over the time-derivative approaches is their use of shadow tomography (the methods of \cite{chen2022exponential} and \cite{huang2020predicting}, respectively) which improves the parallel estimation of many, potentially nonlocal, observables. These algorithms require more sophisticated quantum resources however, needing at least $n$ additional qubits to prepare the (pseudo-)Choi states central to their strategies. Notably, \cite{caro2024learning} requires a quantum memory of size $\Ot(\|H\|^2/\epsilon^2)$ on top of the $n$ Choi qubits, which is inherited from the gentle measurement step of \cite{chen2022exponential}. Meanwhile, \cite{castaneda2023hamiltonian} proposes two approaches to parameter estimation, the first being the classical shadows of Huang, Kueng, and Preskill~\cite{huang2020predicting}. Their second proposal uses the gradient-estimation-based algorithm of \cite{huggins2022nearly}, which achieves better scaling $\ttotal = \Ot(\sqrt{m}/\epsilon)$ but requires a large ancilla register of size $\Ot(m)$. Furthermore, both approaches of \cite{castaneda2023hamiltonian} require the ability to apply $\ctrl{}{e^{\pm \i H t}}$. Either way, both \cite{caro2024learning} and \cite{castaneda2023hamiltonian} are also limited in their structure-learning capabilities, again only giving efficient algorithms if the pool of candidate terms is promised to be polynomially bounded.

Other algorithms of note include that of Yu, Sun, Han, and Yuan~\cite{yu2023robust}, which is based on sparse Pauli channel learning~\cite{flammia2020efficient}. To our knowledge, this is the only prior algorithm which can specifically target sparse, nonlocal Hamiltonians while maintaining efficient scaling. However, their rigorous analysis requires them to assume that the $m$ Pauli terms are uniformly distributed over all possible $\binom{4^n}{m}$ combinations. Their results therefore only hold with high probability over the draw of $H$. The algorithm of Odake, Kristj\'{a}nsson, Soeda, and Murao~\cite{odake2024higher} applies a linear transformation to the Hamiltonian to isolate each term $E_a$, thereby allowing them to learn the parameter $\lambda_a$ via robust phase estimation~\cite{kimmel2015robust}. Iterating this procedure over all known terms furnishes the Heisenberg-limited estimation of $H$. One particular downside to their approach is that their $\tmin = \Omega(\epsilon^2/\|H\|^2)$, where one factor is incurred to approximate the linear transformation, and then another factor is taken to get appropriate guarantees on the robust phase estimation.

Finally, Guti\'{e}rrez~\cite{gutierrez2024simple} describes an algorithm to learn a $k$-local Hamiltonian with some unique properties. First, it is strongly limited to $k = \O(1)$ because of the use of a noncommutative Bohnenblust--Hille inequality~\cite{huang2023learningprocesses,volberg2024noncommutative}. However, this inequality is also what allows it to achieve $\ell_2$-error without explicit dependence on $m$. In contrast, all other attempts to get $\ell_2$-error pay the standard factor of $\sqrt{m}$ for converting from the $\ell_\infty$-norm. Assuming $\|H\| \leq 1$, his algorithm has query complexity $\exp(\O(k^2 + k \log(1/\epsilon)))$. Unfortunately this features rather large constants in the exponents:~unraveling the analysis for unnormalized $\|H\|$ shows that $\ttotal = \O((\|H\|/\epsilon)^{15k + 17})$ and $\tmin = \Omega((\epsilon/\|H\|)^{k + 1})$. The learning algorithm also requires applications of $\ctrl{}{e^{-\i H t}}$, inherited from the use of the quantum Goldreich--Levin algorithm of \cite{montanaro2010quantum} as a subroutine.\\

\textbf{Learning Pauli spectra.} Prior works have also studied the task of learning an operator's Pauli structure, sometimes referred to as its Pauli or Fourier spectrum, given query access to that operator~\cite{atici2007quantum,montanaro2010quantum,chen2023testing,angrisani2023learning,gutierrez2024simple,bakshi2024structure}. These can be understood as quantum analogues to the Goldreich--Levin algorithm, which allows one to efficiently estimate any Fourier coefficient of a Boolean function from queries to that function~\cite{goldreich1989hardcore}. Most works study the Pauli spectrum of a unitary operator, developing algorithms to test and learn unitaries acting on only a constant subsystem of qubits (or more generally, possessing a constant amount of influence over the system)~\cite{atici2007quantum,montanaro2010quantum,chen2023testing,angrisani2023learning}.

Most relevant to the problem of Hamiltonian learning are the algorithms of Bakshi, Liu, Moitra, and Tang~\cite{bakshi2024structure} and Guit\'{e}rrez~\cite{gutierrez2024simple}, which approximate the spectrum of $H$ from that of $e^{-\i H t}$ using series approximations at small $t$. Both results are limited to local Hamiltonians, featuring explicit exponential scalings with respect to the locality of the Pauli terms being learned. In contrast, our algorithm works with an object called the pseudo-Choi state of $H$~\cite{castaneda2023hamiltonian}, which gives direct access to the Pauli spectrum of $H$ rather than deducing it from its time evolution unitary. As such, our approach is not hindered by any locality considerations. At its core, our algorithm is efficient as long as $\|H\| \leq \poly(n)$, and in fact can be applied to learn the Pauli spectrum of \emph{any} operator, not necessarily Hermitian or unitary, as long as one can efficiently block encode it.\\

\textbf{Time reversal in learning problems.} Time reversal is a tool that has been utilized with success in various modern quantum experiments~\cite{garttner2017measuring,mi2021information,braumuller2022probing,colombo2022time,li2023improving}. In learning tasks, the intuitive power of time reversal is its ability to ``refocus'' the evolution of a physical system, enhancing the signal of important many-body correlations that might only manifest at late time scales~\cite{schuster2023learning}. Theoretically, superpolynomial separations have been established between learning protocols with and without access to time-reversal dynamics. Cotler, Schuster, and Mohseni~\cite{cotler2023information} considered the task of distinguishing whether a Haar-random unitary was sampled from $\U(2^n)$ or $\U(2^{n/2}) \otimes \U(2^{n/2})$. Meanwhile, Schuster, Haferkamp, and Huang~\cite{schuster2024random} gave the task of distinguishing between two different classes of pseudorandom quantum circuits. In both cases, it was proven that no time-forward experiment can solve these tasks in polynomial time, but protocols using time-reversed dynamics to probe out-of-time-ordered correlators (OTOCs) can solve them very efficiently.

In our work, we do not use time reversal to probe OTOCs, but rather to implement quantum simulation techniques, namely the quantum singular value transformation (QSVT)~\cite{gilyen2019quantum}. In a similar spirit to the works discussed above, we observe that our time-reversed protocol can nearly reach the Heisenberg limit, while our time-forward algorithm does not surpass the standard quantum limit. However, as we do not consider lower bounds in this paper, we do not establish analogous separations for the sparse Hamiltonian learning problem.\\

\textbf{Indistinguishable particles.} We have restricted attention to many-qubit systems, although we note that there has been some progress on learning bosonic~\cite{li2024heisenberg} and fermionic~\cite{ni2024quantum,mirani2024learning} Hamiltonians. However, these algorithms only consider Hubbard models, and they cannot perform structure learning due to their use of the Hamiltonian reshaping strategy of \cite{huang2023learning}. Conversely, it is interesting to point out that our results can be applied to these systems under an appropriate mapping to qubits. For instance, we immediately get an efficient learning algorithm for any fermionic Hamiltonian with at most polynomially many interaction terms. This is true even if one uses the nonlocal Jordan--Wigner transformation~\cite{jordanwigner} to represent the fermions with qubits, as our techniques have no dependence on the locality of the Pauli terms.

\subsection{Technical overview}

In this section, we will describe the high-level techniques used to build our algorithms and prove our results. As per \cref{prob:1}, we are given an unknown $n$-qubit Hamiltonian $H$ and an integer $m$ with the promise that $H$ is supported on at most $m$ Pauli operators. Therefore, we write $H = \sum_{a=1}^m \lambda_a E_a$, with the understanding that $\lambda \in [-1, 1]^m$ may potentially be padded with zeros. At times, it is natural to consider the size of the Hamiltonian in terms of its operator norm $\|H\|$. Whenever the learner does not have access to this quantity, we keep in mind that the bound $\|H\| \leq m$ can always be used in its stead.

\subsubsection*{Learning with pseudo-Choi states}

The key technique enabling our learning algorithms is the pseudo-Choi state, a concept introduced by Castaneda and Wiebe~\cite{castaneda2023hamiltonian}. They define a version with a reference control qubit, but it is straightforward to consider a referenceless version as well. We will use both variants in our algorithms.

\begin{definition}
    Let $\ket{\Omega} = \frac{1}{\sqrt{2^n}} \sum_{b \in \{0,1\}^n} \ket{b} \ket{b}$. The \emph{pseudo-Choi states} of an $n$-qubit operator $H$, with and without reference, are
    \begin{align}
        \ket{\pchoiref(H)} &= \frac{(H \otimes \I_n) \ket{\Omega} \ket{0} + \ket{\Omega} \ket{1}}{\norm_{\mathrm{ref}}},\\
        \ket{\pchoi(H)} &= \frac{(H \otimes \I_n) \ket{\Omega}}{\norm},\label{eq:refless_pchoi_intro}
    \end{align}
    respectively. $\norm_{\mathrm{ref}}$ and $\norm$ are normalization constants.
\end{definition}

This state draws analogies to the usual Choi state $(\mathcal{E} \otimes \I)(\op{\Omega}{\Omega})$ of a quantum channel $\mathcal{E}$. However, since $H$ does not necessarily induce a quantum channel, its pseudo-Choi state does not exhibit channel--state duality in the sense of the Choi--Jamio{\oldl}kowski isomorphism. Nonetheless, it serves as a powerful learning resource:~\cite{castaneda2023hamiltonian} shows that Clifford classical shadows~\cite{huang2020predicting} can be used to estimate all $m$ parameters of $H$ with $\O(\Delta^2\log(m)/\epsilon^2)$ copies of $\ket{\pchoiref(H/\Delta)}$. Here, $\Delta \geq 2\|H\|$ is a normalization factor necessary for the preparation of this state. This measurement scheme is therefore able to simultaneously learn a large number of observables with no restriction on their locality. We review this protocol in more detail in \cref{sec:pseudo-choi_learning}. Note that the reference control qubit is necessary to learn $\lambda_a$ unambiguously, as the referenceless state is unique only up to an arbitrary nonzero scalar on $H$. (We will explain the utility of this referenceless version later.)

To prepare such states, \cite{castaneda2023hamiltonian} calls upon a subroutine to block encode the Hamiltonian from its time evolution~\cite{low2017hamiltonian,gilyen2019quantum}. This involves block encoding the operator $\sin(H/\Delta) = \frac{1}{2\i}  (e^{\i H/\Delta} - e^{-\i H/\Delta})$ as a linear combination of unitaries (LCU)~\cite{childs2012hamiltonian}, followed by invoking the QSVT~\cite{gilyen2019quantum} to apply a polynomial approximation of the $\arcsin$ function onto that block. The normalization $\|H/\Delta\| \leq 1/2$ ensures that the spectrum being transformed lies in the appropriate domain of the $\arcsin$ approximation. This technique is very accurate:~the polynomial only needs degree $d = \O(\log(1/\varepsilon_{\mathrm{arcsin}}))$ to get spectral error $\varepsilon_{\mathrm{arcsin}}$, and choosing $\varepsilon_{\mathrm{arcsin}} = \Theta(\epsilon/\Delta)$ suffices for the learning problem. Then from $\O(d)$ queries to the controlled time evolution with resolution $|t| = \Theta(1/\Delta)$, one can construct the block encoding of $H/\Delta$. Finally, the pseudo-Choi state $\ket{\pchoiref(H/\Delta)}$ is prepared probabilistically from the block encoding by measuring some ancilla flag qubits;~this succeeds with probability at least $1/2$, so a constant number of repetitions per copy suffices.

The limitations of this approach are two-fold. First, it does not supply an efficient way to learn the terms if they are \emph{a priori} unknown. Second, the state-preparation procedure requires controlled forms of the time evolution, $\ctrl{}{e^{-\i H t}}$ and $\ctrl{}{e^{+\i H t}}$. Access to these resources is unnatural in most learning frameworks, and there are no-go conditions for controlling black-box unitaries~\cite{araujo2014quantum,gavorova2024topological}. Furthermore, in the time-forward access model one cannot directly query the (controlled) inverses either. In order to make the pseudo-Choi state method applicable to more conventional learning scenarios, we seek a way to lift these limitations.

\subsubsection*{Preparing pseudo-Choi states without controlled evolutions}

To get our main results \cref{thm:tr_learning,thm:tf_learning}, we show how to prepare pseudo-Choi states without either controlled time evolutions or (in the second access model) time reversal. The guarantees on our state-preparation algorithms are as follows.

\begin{theorem}[Pseudo-Choi state preparation with uncontrolled time-reversal queries]\label{thm:tr_state_prep}
    Let $\Delta \geq 2\|H\|$, $\epsilon > 0$. From access to time-reversal dynamics, we can prepare $\O(\epsilon/\Delta)$-close states (in trace norm) to $\ket{\pchoiref(H/\Delta)}$ and $\ket{\pchoi(H)}$ with success probabilities at least $1/2$ and $\Omega(\|\lambda/\Delta\|^2)$, respectively. For either state, each preparation attempt uses $\O(\frac{1}{\Delta} \log(\Delta/\epsilon))$ evolution time to $e^{-\i H t}$ with time resolution $|t| \geq \tmin = \Omegat(\epsilon/\Delta^2)$. The number of ancilla qubits required is at most $\nanc = n + 3$.
\end{theorem}

\begin{theorem}[Pseudo-Choi state preparation with uncontrolled time-forward queries]\label{thm:tf_state_prep}
    Let $\Delta \geq 2\|H\|$, $\epsilon > 0$. From access to only time-forward dynamics, we can prepare $\O(\epsilon/\Delta)$-close states (in trace norm) to $\ket{\pchoiref(H/(\Lambda\Delta))}$ and $\ket{\pchoi(H)}$ with success probabilities at least $1/2$ and $\Omega(\epsilon^2 \|\lambda/\Delta^2\|^2)$, respectively. The parameter $\Lambda$ is known analytically and obeys $\Lambda = \O(\Delta/\epsilon)$. For either state, each preparation attempt uses $\O(\frac{1}{\Delta} \log^2(\Delta/\epsilon))$ evolution time to $e^{-\i H t}$ with time resolution $t \geq \tmin = \Omegat(\epsilon/\Delta^2)$. The number of ancilla qubits required is $\nanc = n + \O(\log\log(\Delta/\epsilon))$.
\end{theorem}

Observe that the success probability for $\ket{\pchoiref(H/\Delta)}$ is always at least $1/2$. Meanwhile for $\ket{\pchoi(H)}$, the dependence of $\|\lambda\|^2$ will cancel nicely with its use case (which we describe later). The demand for $\O(\epsilon/\Delta)$ accuracy is a consequence of the fact that $\ket{\pchoiref(H/\Delta)}$ encodes the normalized parameters $\lambda_a/\Delta$. In order to get the desired in accuracy in learning $\lambda_a$, we need the state-preparation error to be $\leq c\epsilon/\Delta$ for a small constant $c \ll 1$. Below, we describe the techniques that make these constructions possible.

\subsubsection*{Controlization to construct block encodings}

We remove the need for controlled time evolutions through a technique called controlization. While various approaches for controlization have been proposed~\cite{janzing2002quantum,zhou2011adding,friis2014implementing,nakayama2015quantum,dong2019controlled,chowdhury2024controlization}, we employ the one described by Odake, Kristj\'{a}nsson, Taranto, and Murao~\cite{odake2023universal} due to its simplicity and universal applicability. Let $H_0 = H - \tr(H)\I_n/2^n$ be the traceless part of $H$. The subroutine generates a mixed-unitary channel $\mathcal{E}(t)$ which approximates $\ctrl{}{e^{-\i H_0 t}}$ to any target diamond-norm distance $\gamma$. This approach to controlization also has the nice feature that it exhibits zero overhead in terms of total evolution time:~in order to approximate $\ctrl{}{e^{-\i H_0 t}}$, we query $e^{-\i H \delta t}$ for a total of $N = t/\delta t$ times. The downside is that we require a small time resolution of $\tmin = t/N$ and $\O(n)$ quantum gates per query, so our main task is bounding what $N$ suffices to achieve the total target error.

The controlization technique of \cite{odake2023universal} is based on Pauli twirling and randomized product formulas. It is well know that averaging over the Pauli group gives the identity
\begin{equation}
    \frac{1}{4^n} \sum_{P \in \Pauli{n}} P^\dagger H P = \frac{\tr H}{2^n} \I_n.
\end{equation}
This can be extended by appending a control qubit and replacing each Pauli gate by $\ctrl{0}{P}$. In this case the identity becomes
\begin{equation}
     \frac{1}{4^n} \sum_{P \in \Pauli{n}} \ctrl{0}{P^\dagger} (\I_1 \otimes H) \ctrl{0}{P} = \op{1}{1} \otimes H_0 + \frac{\tr H}{2^n} \I_{n+1}.
\end{equation}
Exponentiating this twirled Hamiltonian gives, up to some global phase, the unitary $\ctrl{1}{e^{-\i H_0 t}}$. We then straightforwardly extend this idea to synthesize multi-controlled unitaries $\ctrl{b}{e^{-\i H_0 t}}$, where $b \in \{0, 1\}^k$ is any $k$-qubit control string. Note that each such controlled Pauli gate can be implemented with $\O(k + n)$ single- and two-qubit quantum gates.

In order to evolve by the twirled Hamiltonian, Trotter's product formula tells us that $e^{-\i \sum_j h_j} \approx \prod_j e^{-\i h_j}$. It is enticing to simply concatenate the unitaries $e^{-\i h_j} = \ctrl{0}{P^\dagger} e^{-\i H t / 4^n} \ctrl{0}{P}$, but clearly this is inefficient;~even ignoring the Trotter error, there are exponentially many terms. Fortunately, Campbell's qDRIFT protocol~\cite{campbell2019random} was designed with precisely this concern in mind. Rather than include each term in the simulation, qDRIFT constructs random product formulas wherein terms are selected with probability proportional to their norm. In our context, all terms have the same weight, so we just sample from $\Pauli{n}$ uniformly. By treating this ensemble of random product formulas as a mixed-unitary channel (wherein we draw a new formula for each instance), the analysis of \cite{campbell2019random} (and tightened in \cite{chen2021concentration}) shows that qDRIFT sequences of length $N = \O(t^2 \|H\|^2/\gamma)$ suffice to get diamond-norm distance at most $\gamma$. Since we only need applications running for time $t = \O(1/\Delta) = \O(1/\|H\|)$, this simplifies to $N = \O(1/\gamma)$.

To determine what $\gamma$ we need to maintain our learning guarantees, note that our approximation prepares a mixed state from which we probabilistically prepare the pseudo-Choi states. A careful error analysis shows that $\O(\epsilon/\Delta)$ error in trace distance of the mixed states (hence in diamond-norm distance of the channels) suffices. Each block-encoding circuit makes $Q_1 = \O(\log(\Delta/\epsilon))$ queries to $\ctrl{}{e^{\pm \i H/\Delta}}$, so choosing $\gamma = \O(\epsilon/(Q_1 \Delta)) = \Ot(\epsilon/\Delta)$ nets the claim of \cref{thm:tr_state_prep}. (For the success probabilities, observe that $\Pr(\text{Preparing } \ket{\pchoiref(H/\Delta)}) \geq 1/2$ always holds because of the reference term, while $\Pr(\text{Preparing } \ket{\pchoi(H)})$ is proportional to the norm of the block-encoded matrix.)

\subsubsection*{Block encoding the matrix logarithm without time reversal}

The construction above removes the need for $\ctrl{}{e^{\pm \i H t}}$, but still requires negative-time queries in order to implement the QSVT (and even the LCU for $\sin(H/\Delta)$). To develop a completely time-reversal-free approach, we construct an entirely different quantum circuit for block encoding $H/\Delta$. The basic idea is simple:~we appeal to the well-known power series for the (matrix) logarithm~\cite{hall2015lie}.

For any square complex matrix $X$ such that $\|X\| < \log 2$, the following identity holds:
\begin{equation}\label{eq:log_eX}
    X = \sum_{k=1}^\infty (-1)^{k+1} \frac{(e^X - \I)^k}{k}.
\end{equation}
Letting $X = -\i H/\Delta$, this formula gives us a relationship between the Hamiltonian and a linear combination of its time evolution operators which only run forward in time. Block encoding the LCU still requires controlled unitaries, but by our prior discussion we now have a technique to approximate them. Note that our normalization requirement is slightly relaxed compared to the QSVT approach, but for sake of clarity let us impose $\Delta \geq 2\|H\|$ as before.\footnote{We can actually relax the QSVT normalization further:~one can allow $\|H/\Delta\| \leq 1 - \delta$ for any $\delta \in (0, 1)$ by increasing the polynomial degree to $d = \O(\log(\Delta/\epsilon)/\delta)$~\cite[Lemma 70]{gilyen2019quantum}.}

Technically, \cref{eq:log_eX} is not written in the form of an LCU. It is also a series with an infinite number of terms. Our analysis is then to determine where we should truncate the series to maintain sufficient accuracy, and how to express the truncated sum as a proper LCU. The former is straightforward:~since $\|X\| \leq \frac{1}{2}$, we have $\|e^X - \I\| \leq \frac{1}{2}$. Let $L_K$ be the partial sum of \cref{eq:log_eX} up to $k \leq K$ terms. An application of the triangle inequality gets $\|X - L_K\| \leq 2^{-(K+1)}$, which means that $K = \lceil \log_2(\frac{1}{2\gamma}) \rceil$ ensures a truncation error of at most $\gamma$.

To express $L_K$ as an LCU, we expand the terms $(e^{-\i H/\Delta} - \I)^k$ using the binomial theorem. With the help of combinatorial identities, we can derive analytical expressions for coefficients $c_j$ such that $L_K = \sum_{j=0}^K c_j e^{-\i (H/\Delta) j}$. Techniques for implementing LCUs are well established~\cite{childs2012hamiltonian,gilyen2019quantum}, requiring controlled instances of the unitaries on only $\lceil \log_2(K+1) \rceil = \O(\log\log(1/\gamma))$ control qubits. The bottleneck of any LCU implementation is the size of its coefficient $\ell_1$-norm, $\Lambda = \sum_j |c_j|$. Using the analytical form of the $c_j$'s, we can compute an essentially tight bound of $\Lambda = \O(2^K)$. Recalling that we need to take error below $\gamma = \O(\epsilon/\Delta)$, we immediately get  $\Lambda = \O(\Delta^2/\epsilon^2)$ and \cref{thm:tf_state_prep} follows. For learning the Hamiltonian parameters, our LCU now prepares the state $\ket{\pchoiref(H/(\Lambda\Delta))}$ (still with at least $1/2$ probability), which requires us to refine the learning precision by $1/\Lambda$. In turn, this raises the copy complexity to $\Lambda^2 \cdot \O(\Delta^2 \log(m)/\epsilon^2) = \Ot(\Delta^4 / \epsilon^4)$. Each copy costs $\Ot(1/\Delta)$ evolution time to produce, leading to the claimed result in \cref{thm:tf_learning}.

\subsubsection*{Identifying Pauli spectra using pseudo-Choi states}

So far we have discussed the utility of the pseudo-Choi state with reference, as introduced by \cite{castaneda2023hamiltonian}. Let us now turn to the referenceless version as in \cref{eq:refless_pchoi_intro}. We invoke it in order to accomplish the structure identification task. The reference was essential to learn values of the Hamiltonian coefficients $\lambda_a$, but it serves no purpose when we only care about detecting which coefficients are large relative to the others. By performing Bell measurements~\cite{montanaro2017learning} on copies of $\ket{\pchoi(H)}$, we can learn which terms in $H$ have weight greater than $\epsilon$.

\begin{theorem}[Bell sampling for structure learning]\label{thm:bell_sampling_alg_summary}
    With high probability, we can generate a set of Pauli labels containing $S_\epsilon(H) = \{E_a : |\lambda_a| > \epsilon\}$. This process consumes $\O(\|\lambda\|^2 \log(m)/\epsilon^2)$ copies of $\ket{\pchoi(H)}$, measuring each in the Bell basis.
\end{theorem}

Observe that while we need $N = \Ot(\|\lambda\|^2/\epsilon^2)$ copies, the probability of producing a single copy is $p = \Omega(\|\lambda\|^2/\Delta^2)$ (time-reversal) or $p = \Omega(\epsilon^2\|\lambda\|^2/\Delta^4)$ (time-forward) from $\Ot(1/\Delta)$ evolution time. It is a standard fact that $\O(N/p)$ Bernoulli trials are sufficient to get $N$ successes. Hence, $\Ot(\Delta/\epsilon^2)$ (time-reversal) or $\Ot(\Delta^3/\epsilon^4)$ (time-forward) total time evolution is used to learn the structure of $H$. Unlike prior approaches to estimating Pauli spectra of unitary operators~\cite{montanaro2010quantum,angrisani2023learning,gutierrez2024simple,bakshi2024structure}, our algorithm can in principle be applied to any $n$-qubit operator $A$, as long as one can consistently generate copies of its pseudo-Choi state $\ket{\pchoi(A)}$. We are also not restricted by any locality considerations, and the total number of queries to $e^{-\i H t}$ essentially only depends on $\epsilon/\|H\|$ (i.e., the target error in ``natural'' units of the Hamiltonian's norm).

\cref{alg:pauli_spectrum_learning} describes the core idea of our structure-identifying algorithm. For each copy of $\ket{\pchoi(H)}$, we measure in the $2n$-qubit Bell basis, which can be described as $\{(P \otimes \I_n) \ket{\Omega} : P \in \Pauli{n}\}$. In turn, this induces a distribution governed by $\Pr(E_a) = |\lambda_a|^2 / \|\lambda\|^2$. Sampling from this distribution allows us to learn the Pauli terms in $H$ that have relatively high weight.

\begin{algorithm}
    \caption{\label{alg:pauli_spectrum_learning}Identifying the support of the Pauli spectrum of a Hamiltonian.}
    \KwIn{$N$ copies of the referenceless pseudo-Choi state $\ket{\pchoi(H)}$.}

    \BlankLine

    Let $\widehat{S} \gets \varnothing$\;
    
    \For{$j = 1, \ldots, N$}{
        Measure $\ket{\pchoi(H)}$ in the Bell basis $\{\ket{P} = (P \otimes \I_n) \ket{\Omega} : P \in \Pauli{n}\}$\;
        Let $\ket{P^{(j)}}$ be the observed outcome\;
        $\widehat{S} \gets \widehat{S} \cup \{P^{(j)}\}$\;
    }
    \Return{$\widehat{S}$}
\end{algorithm}

To analyze the cost of this algorithm, we need to bound the number of measurements such that we observe at least all the terms with $|\lambda_a| > \epsilon$. This is a variant of the coupon collector problem, 
which is well studied~\cite{shioda2007some}. Our setup has two modifications from the classic version of the problem. First, we have a nonuniform distribution, which bottlenecks the sample complexity based on the least likely coupon. Second, we only care to collect low-rarity coupons (i.e., those such that $|\lambda_a| > \epsilon$). Hence the rarest coupon we need to collect has at least $\Pr(E_a) > \epsilon^2/\|\lambda\|^2$. This condition allows us to bound the tail of the cumulative distribution where, after $N$ i.i.d.~draws of the coupons (Pauli terms), we still have not observed at least one of each type. We show that this probability is at most $\delta$ whenever $N = \O(\|\lambda\|^2 \log(m/\delta)/\epsilon^2)$.

Besides this main argument, there are some smaller technical details to take care of. First is the fact that we do not prepare $\ket{\pchoi(H)}$ exactly, but rather an approximation $\ket{\pchoi(\widetilde{H})}$. Suppose that $\|H - \widetilde{H}\| \leq \varepsilon$ for some small $\varepsilon$, which we can tune by constructing a more precise block encoding. Then, we show that $S_\epsilon(H) \subseteq S_{\epsilon+\varepsilon}(\widetilde{H})$. While we may pick up terms with $|\lambda_a| \leq \epsilon$ (this is possible even if sampling from the exact pseudo-Choi state), it is fine to learn a superset. This is because the next stage predicts all $\lambda_a$ using classical shadows, which only has an explicit logarithmic cost in the number of parameters. Any parameters deemed too small can simply be discarded. Clearly since we only take a polynomial number of samples, this superset has bounded size. But we can prove something even stronger:~$S_{\epsilon+\varepsilon}(\widetilde{H}) \subseteq \{E_1, \ldots, E_m\}$. That is, even if the approximation $\widetilde{H}$ has support that leaks onto other Pauli terms, its $(\epsilon + \varepsilon)$-support never contains those erroneous terms.

Finally, we remark that the controlization approximation affects this algorithm more severely than the parameter estimation stage. We find that the Hamiltonian coefficients become perturbed on the order of $\Delta \sqrt{Q_1 \gamma}$, where $Q_1$ is the number of queries to the time evolution for one block encoding and $\gamma$ is the controlization error per query. In order to ensure the correctness of the Bell sampling procedure, we need to adjust the $\epsilon$ threshold of \cref{thm:bell_sampling_alg_summary} and choose a sufficiently small $\gamma$. The upshot of this analysis is a time resolution of $\tmin = \Omegat(\epsilon/\Delta^3)$ under the time-reversal access model, and $\tmin = \Omegat(\epsilon^4/\Delta^5)$ under the conventional access model, for this stage of the algorithm. As is inherent to controlization, $\ttotal$ and $\Nexp$ are unaffected.

\subsubsection*{Bootstrapping to (nearly) Heisenberg-limited scaling}

Having described an algorithm that learns at the standard quantum limit, we consider how to the enhance the procedure to reach the Heisenberg limit. Using the pseudo-Choi state approach, this is ultimately achieved by a technique called uniform spectral amplification~\cite{low2017hamiltonian}, a block-encoding generalization of amplitude amplification. As such, we will need access to time evolution inverses, so the following discussion is only viable under the time-reversal model.

To achieve Heisenberg-limited scaling, we use the precision bootstrapping framework, an idea that has enjoyed many fruitful successes in the dynamical learning paradigm~\cite{kimmel2015robust,haah2023query,dutkiewicz2023advantage,zhao2023learning,bakshi2024structure}. The first use of the strategy for many-body Hamiltonian learning is attributed to Dutkiewicz, O’Brien, and Schuster~\cite{dutkiewicz2023advantage}, inspired by the feedback control scheme of Kura and Ueda~\cite{kura2018finite} for qudit learning. Its broader applicability was subsequently appreciated by \cite{bakshi2024structure}. We briefly review the idea here.

Let $\alg : (H, \epsilon, \delta) \mapsto \widehat{\lambda} \in \R^m$ be some base learning algorithm with the guarantee that $\Pr(\|\lambda - \widehat{\lambda}\|_\infty \leq \epsilon) \geq 1 - \delta$. The idea is to bootstrap applications of this base algorithm, with constant precision $\epsilon_0 = 1/2$, to the Heisenberg limit $\ttotal \sim 1/\epsilon$ by adaptively adjusting the Hamiltonian being learned via quantum control. Consider a previously learned estimate $\widehat{H}$ and form a ``residual Hamiltonian'' $R = H - \widehat{H}$. Assuming that the current estimate has error $\leq \eta$, we can amplify $R$ by a factor of $1/\eta$ to boost those residual errors up to a constant. Then if we learn $R/\eta$ using $\alg$ with merely constant precision $\epsilon_0 = 1/2$, we can improve our estimate to $\eta/2$ error. Recursing on this process for multiple rounds $j = 0, 1, \ldots, T$ for some $T$, it is clear that each round produces an estimate with error $\eta_j = 2^{-j}$. Hence $T = \lfloor \log_2(1/\epsilon)\rfloor$ suffices to get the target accuracy $\epsilon$. Taking some care with the acceptable amount of failure probability establishes the correctness of the algorithm. This overall strategy is outlined in \cref{alg:DOS_bootstrap}.

\begin{algorithm}
    \caption{\label{alg:DOS_bootstrap}Bootstrapped learning algorithm with Heisenberg-limited scaling.}
    \KwIn{Error parameters $\epsilon, \delta \in (0, 1)$, base learning algorithm $\alg$, and black-box access to time evolution under the unknown Hamiltonian $H = \sum_{a=1}^m \lambda_a E_a$.}
    \KwOut{An estimate $\widehat{\lambda}$ such that $\Pr(\|\lambda - \widehat{\lambda}\|_\infty \leq \epsilon) \geq 1 - \delta$.}

    \BlankLine

    Let $\widehat{\lambda}^{(0)} \gets (0, \ldots, 0) \in \R^m$\;
    Let $T \gets \lfloor \log_2(1/\epsilon) \rfloor$\;
    
    \For{$j = 0, 1, \ldots, T$}{
        Let $\zeta_j \gets \delta / 2^{T+1-j}$\;
        Let $\eta_j \gets 1/2^j$\;
        Let $\widehat{H}_j \gets \sum_{a=1}^m \widehat{\lambda}^{(j)}_a E_a$\;
        Let $R_j \gets H - \widehat{H}_j$\;
        $\widehat{r}^{(j)} \gets \alg(R_j/\eta_j, 1/2, \zeta_j)$\;
        Let $\widehat{\lambda}^{(j+1)} \gets \widehat{\lambda}^{(j)} + \eta_j \widehat{r}^{(j)}$\;
    }
    \Return{$\widehat{\lambda} \gets \widehat{\lambda}^{(T+1)}$}
\end{algorithm}

To achieve Heisenberg-limited scaling, the process of boosting the residual by $1/\eta$ needs to take $\O(\eta)$ time, otherwise the advantage is lost. In the conventional time-evolution setting, this holds because applying $e^{-\i R/\eta}$ takes time $1/\eta$ by definition. In \cite{dutkiewicz2023advantage}, they consider \emph{continuous} quantum control, wherein the experimenter can implement $e^{-\i (H - \widehat{H}) t}$ directly, while \cite{bakshi2024structure} assumes only discrete control, requiring them to approximate the residual evolution via Trotterization. In our algorithm, we prepare pseudo-Choi states from a block encoding of the Hamiltonian. Thus the bootstrapping strategy requires us to block encode the residual Hamiltonian, which comes with unique advantages and challenges. The main advantage is that we can block encode $-\widehat{H}$ exactly as an LCU~\cite{childs2012hamiltonian}, thereby avoiding the complications of Trotter error in the discrete control model.

On the other hand, boosting the block-encoded residual by $1/\eta$ is more challenging. To do this, we invoke Low and Chuang's uniform spectral amplification technique~\cite{low2017hamiltonian},\footnote{A generalization for amplifying the singular values of a not-necessarily-Hermitian matrix~\cite[Theorem 30]{gilyen2019quantum} also accomplishes this task.} which essentially allows us to map $R \mapsto R/\eta$, up to some controllable error $\varepsilon_{\mathrm{amp}}$, as long as $\|R/\eta\| \leq 1$. This costs $\O(\log(1/\varepsilon_{\mathrm{amp}}) / \eta)$ queries to the block encoding of $R = H - \widehat{H}$. Meanwhile, recall that block encoding $H$ itself costs $\O(\log(1/\varepsilon_{\mathrm{arcsin}}))$ queries to $\ctrl{}{e^{\pm\i H/\Delta}}$. A careful error analysis shows that we need to take $\varepsilon_{\mathrm{arcsin}} = \Theta(\eta/\Delta)$ (on the other hand, $\varepsilon_{\mathrm{amp}} = \Theta(\epsilon_0/\Delta)$ suffices). Thus in terms of $\eta$, the cost of each round of the bootstrap scales as $\log(1/\eta)/\eta$, giving us nearly Heisenberg-limited scaling up to the factor of $\log(1/\eta) \leq \log(1/\epsilon)$.

Note that this argument applies to the structure identification stage of our algorithm as well. For intuition here, consider the first round $j = 0$, wherein we identify all the terms which have $|\lambda_a| > 1/2$ and estimate them, giving us $\widehat{H}_1$. In the next round $j = 1$, we synthesize $H - \widehat{H}$ and amplify by $2$. This splits the Hamiltonian terms into two sets:~those in $S_{1/2}(H)$, and those in the complement $S_{1/2}(H)^\comp$. We only care to identify the elements of $S_{1/2}(H)^\comp$;~after amplification, all the such terms larger than $1/4$ become larger than $1/2$, which again can be identified with just a constant $1/2$ precision. Iterating in this way, we can deduce the identity of all terms larger than $\epsilon$ with only $\lfloor \log_2(1/\epsilon) \rfloor$ rounds. Combining this with the parameter estimation procedure gives us the claim in \cref{thm:tr_learning}.

One final detail to point out is the normalization of the residual Hamiltonian. To guarantee that $\|R/\eta\| \leq 1$, we make the particular choice $\Delta = 2m$. This is because our error metric is with respect to the vector-norm of the coefficients, which typically has a loose conversion from the operator norm:
\begin{equation}
    \|H - \widehat{H}\| \leq \|\lambda - \widehat{\lambda}\|_1 \leq m \|\lambda - \widehat{\lambda}\|_\infty.
\end{equation}
That said, there exist Hamiltonians which saturate this inequality, so in the worst case we cannot expect to tighten it. While we use the $\ell_\infty$-norm as our metric, converting to different norms simply incurs factors of $m$ in the sample complexity.

\subsection{Discussion}

In this paper, we have described algorithms to efficiently learn sparse, nonlocal Hamiltonians from minimal prior information. This addresses an open direction posed in \cite{bakshi2024structure}. At the center of our results is the pseudo-Choi state introduced by \cite{castaneda2023hamiltonian}, which enables accurate and parallelized learning of many Hamiltonians terms without locality restrictions or a large quantum memory. The original proposal required controlled time evolutions and time reversal to prepare these states;~we have shown how to lift these resource requirements, thereby bringing the scope of the problem down to conventional black-box query scenarios.

For the time-reversal access model, our algorithm can achieve nearly Heisenberg-limited scaling with the pseudo-Choi method. For the time-forward access model, we prepare pseudo-Choi states using a strategy for block encoding the logarithm of a unitary operator without the QSVT. These constructions may be of independent interest more broadly.

To learn the structure of a sparse, nonlocal Hamiltonian, we established a protocol to estimate the Pauli spectrum of any Hamiltonian by sampling its (referenceless) pseudo-Choi state. Parallel to other works which establish quantum analogues of the Goldreich--Levin algorithm for unitaries~\cite{montanaro2010quantum,angrisani2023learning} and local Hamiltonians~\cite{gutierrez2024simple,bakshi2024structure}, our result establishes a version for nonlocal Hamiltonians. In fact, this algorithm is applicable not only to Hamiltonians but to any matrix that we can efficiently block encode. It would be interesting to find applications of this result in other contexts.

While our work affirms that it is possible to efficiently learn any Pauli-sparse Hamiltonian from queries to its time evolution, many important questions remain open. We list a few below.

\begin{enumerate}
    \item Is there an information-theoretic separation between Hamiltonian learning protocols with and without time reversal? While exponential separations have been established in unitary learning problems~\cite{cotler2023information,schuster2024random}, these results do not immediately translate to our sparse Hamiltonian learning problem. Since we have established that this problem can be solved in polynomial time, any separation can only be polynomially large.

    \item Can the time resolution of our algorithms be improved? Because the small resolution arises from controlization, an access model given $\ctrl{}{e^{-\i H t}}$ immediately bypasses this issue. Within the controlization perspective, an analysis of randomized block encodings (e.g., as in \cite{martyn2024halving}) might be useful.

    \item What is the asymptotically tightest lower bound for the evolution time needed to solve \cref{prob:1}? To our knowledge, the current best lower bound for $\ell_\infty$-error $\epsilon$ is $\Omega(1/\epsilon)$, obtained by Huang, Tong, Fang, and Su~\cite{huang2023learning} by reducing the $m$-term problem to a distinguishing task between two single-term Hamiltonians. For $\ell_1$-error $\epsilon_1$, Ma, Flammia, Preskill, and Tong~\cite{ma2024learning} proved that $\Omega(m/\epsilon_1)$ evolution time is necessary. Meanwhile for non-sparse Hamiltonians ($m \leq 4^n$), Bluhm, Caro, and Oufkir~\cite{bluhm2024hamiltonian} showed exponential lower bounds for protocols with ($\Omega(4^n/n)$) and without ($\Omega(4^n/\epsilon)$) ancillas and adaptivity. A related problem is learning arbitrary unitary channels, whose query lower bound of $\Omega(4^n/\epsilon)$ was given by Haah, Kothari, O'Donnell, and Tang~\cite{haah2023query}.

    \item To what extent is our algorithm robust to state-preparation and measurement errors?
\end{enumerate}

%% file: background.tex
\section{Background}

\subsection{Notation}

We denote the imaginary unit by $\i \equiv \sqrt{-1}$. For any integer $k \geq 1$, define the set $[k] \coloneqq \{1, \ldots, k\}$. The natural logarithm is denoted by $\log$, and other bases will be specified when necessary. Asymptotic upper and lower bounds are denoted by $\O(\cdot)$ and $\Omega(\cdot)$, respectively. We write $f = \Theta(g)$ if $f = \O(g)$ and $f = \Omega(g)$. Tildes on big-O notation suppress polylogarithmic factors, e.g., $\Ot(f) \coloneqq \O(f \polylog(f))$. When we only care that a quantity is polynomially bounded, we write $\poly(f, g, \ldots) \coloneqq (fg \cdots)^{\O(1)}$. Generally speaking, for an object $z$, we will use $\widehat{z}$ to denote an estimate of it and $\widetilde{z}$ to denote an approximation of it.

\subsection{Linear algebra}

We collect some standard facts that we will make regular use of.

\begin{proposition}[Vector norm conversion]\label{prop:vector_norm_conversion}
    Let $\|x\|_p \coloneqq \l( \sum_{j=1}^d |x_j|^p \r{)^{1/p}}$ be the $\ell_p$-norm of any $x \in \C^d$, with $\|x\|_\infty \coloneqq \lim_{p\to\infty}\|x\|_p = \max_{j \in [d]} |x_j|$. For all $0 < r < p$, these vector norms obey
    \begin{equation}
        \|x\|_p \leq \|x\|_r \leq d^{\frac{1}{r} - \frac{1}{p}} \|x\|_p.
    \end{equation}
    These bounds also hold in the limit $p \to \infty$, in the sense that $\|x\|_\infty \leq \|x\|_r \leq d^{\frac{1}{r}} \|x\|_\infty$.
\end{proposition}

The $\ell_2$-norm is particularly important, and when the context is clear, we may simply write it as $\|x\|$.

\begin{definition}
    For any matrix $A \in \C^{d \times d}$, the \emph{Schatten $p$-norm} $\|A\|_p$ is the $\ell_p$-norm of its singular values.
\end{definition}

As such, Schatten norms also obey the bounds of \cref{prop:vector_norm_conversion}. Important matrix norms are the operator norm $\|A\| \equiv \|A\|_\infty$ and the Frobenius norm $\|A\|_F \equiv \|A\|_2$, which are related by $\|A\| \leq \|A\|_F \leq \sqrt{d} \|A\|$. Schatten norms also appear in a matrix version of H\"{o}lder's inequality.

\begin{proposition}[H\"{o}lder's inequality]
    Let $A, B \in \C^{d \times d}$. For any $1 \leq p, q \leq \infty$ such that $\frac{1}{p} + \frac{1}{q} = 1$,
    \begin{equation}
        |{\tr(A^\dagger B)}| \leq \|A\|_p \|B\|_q.
    \end{equation}
\end{proposition}

Next we consider a Hilbert space $(\C^2)^{\otimes n} \cong \C^{2^n}$ of $n$ qubits. Let $\I_n \in \C^{2^n \times 2^n}$ be the $n$-qubit identity operator. We will write $\I$ when the dimension is clear from context.

\begin{definition}
The \emph{Pauli matrices} are defined as
\begin{equation}
    \I = \begin{pmatrix}
        1 & 0\\
        0 & 1
    \end{pmatrix}, \quad
    X = \begin{pmatrix}
        0 & 1\\
        1 & 0
    \end{pmatrix}, \quad
    Y = \begin{pmatrix}
        0 & -\i\\
        \i & 0
    \end{pmatrix}, \quad
    Z = \begin{pmatrix}
        1 & 0\\
        0 & -1
    \end{pmatrix}.
\end{equation}
The set of $n$-qubit Pauli operators $\Pauli{n} \coloneqq \{\I, X, Y, Z\}^{\otimes n}$ is a complete orthogonal operator basis for $\C^{2^n}$, obeying the trace-orthogonality condition $\tr(P^\dagger Q) = 2^n \delta_{PQ}$ for all $P, Q \in \Pauli{n}$.
\end{definition}

As such, any $n$-qubit operator $A$ can be written as
\begin{equation}
    A = \sum_{P \in \Pauli{n}} x_P P
\end{equation}
where $x_P = \frac{1}{2^n} \tr(P^\dagger A) \in \C$ are sometimes referred to as the \emph{Fourier coefficients} or \emph{Pauli spectrum} of $A$. There is a close relation between the $\ell_2$-norm of these coefficients, the Frobenius norm, and the norm of Choi-like vectors of $A$.

\begin{proposition}[Fourier coefficient identities]
    Let $A = \sum_{P \in \Pauli{n}} x_P P$ be any $2^n \times 2^n$ complex matrix. Parseval's identity states that
    \begin{equation}
        \|x\|_2 = \frac{1}{\sqrt{2^n}} \|A\|_F.
    \end{equation}
    An important related identity is
    \begin{equation}
        \frac{1}{\sqrt{2^n}} \|A\|_F = \| (A \otimes \I_n) \ket{\Omega} \|,
    \end{equation}
    where $\ket{\Omega} = \frac{1}{\sqrt{2^n}} \sum_{b \in \{0, 1\}^n} \ket{b} \ket{b}$. This is a consequence of a more general result,
    \begin{equation}
        \ev{(A \otimes \I_n)}{\Omega} = \frac{1}{2^n} \tr A.
    \end{equation}
\end{proposition}

We also standardize our notation for controlled unitaries.

\begin{definition}
    For any bit string $b \in \{0, 1\}^k$, the $k$-qubit controlled variant of an $n$-qubit unitary operator $U$ is
    \begin{equation}
        \ctrl{b}{U} \coloneqq \op{b}{b} \otimes U + (\I_k - \op{b}{b}) \otimes \I_n.
    \end{equation}
    When it is unimportant to specify the control bits $b$, or it is otherwise clear from context, we will simply write $\ctrl{}{U}$.
\end{definition}

Finally, we also work with \emph{quantum channels}:~completely positive, trace-preserving maps on the space of operators. The natural norm to equip quantum channels with is the diamond norm, also known as the completely bounded trace norm.

\begin{definition}
    Let $\mathcal{E} : \C^{d \times d} \to \C^{d \times d}$ be a quantum channel. The \emph{diamond norm} of $\mathcal{E}$ is
    \begin{equation}
        \| \mathcal{E} \|_\diamond \coloneqq \sup\{ \| (\mathcal{E} \otimes \I)(\rho) \|_1 : \rho \in \C^{d^2 \times d^2}, \|\rho\|_1 \leq 1\}.
    \end{equation}
\end{definition}

\subsection{Block encodings}

Here we collect some standard results about block encodings, namely the ability to take linear combinations of unitaries. Our presentation follows that of \cite{gilyen2019quantum} for convenience.

\begin{definition}[{\cite[Definition 43]{gilyen2019quantum}}]
    Let $A$ be an $n$-qubit operator and $\alpha, \varepsilon \geq 0$. An $(n + a)$-qubit unitary operator $B$ is called an $(\alpha, a, \varepsilon)$-\emph{block encoding} of $A$ if
    \begin{equation}
        \| A - \alpha (\bra{0^a} \otimes \I_n) B (\ket{0^a} \otimes \I_n) \| \leq \varepsilon.
    \end{equation}
\end{definition}

For brevity, we will use the shorthand:~$(\alpha, a, \varepsilon)$-BE. One may write an $(\alpha, a, 0)$-BE as
\begin{equation}
    B = \begin{pmatrix}
        A/\alpha & \cdot\ \\
        \cdot & \cdot\ 
    \end{pmatrix}
\end{equation}
when the unspecified blocks are unimportant.

\begin{definition}[{\cite[Definition 51]{gilyen2019quantum}}]
    Let $y \in \C^K$ with $\|y\|_1 \leq \beta$, and $\varepsilon \geq 0$. The pair of $b$-qubit unitaries $(\prep_L, \prep_R)$ is called a $(\beta, b, \varepsilon)$-\emph{state-preparation pair} for $y$ if
    \begin{equation}
        \sum_{j=0}^{K-1} |\beta(c_j^* d_j) - y_j| \leq \varepsilon,
    \end{equation}
    where $c_j = \melem{j}{\prep_L}{0^b}$ and $d_j = \melem{j}{\prep_R}{0^b}$, and for all $j = K, \ldots, 2^b - 1$, $c_j^* d_j = 0$.
\end{definition}

\begin{proposition}[{Linear combination of block encodings~\cite[Lemma 52]{gilyen2019quantum}}]\label{prop:LCBE}
    Let $A = \sum_{j=0}^{K-1} y_j A_j$ be an $n$-qubit operator, expressed as a linear combination of $K$ operators $A_j$ with $y \in \C^K$. Let $\varepsilon_1, \varepsilon_2 \geq 0$. For each $A_j$, suppose that there is an $(\alpha, a, \varepsilon_2)$-block encoding $B_j$. Also, let $(\prep_L, \prep_R)$ be a $(\beta, b, \varepsilon_1)$-state-preparation pair for $y$. Define the $(n + a + b)$-qubit unitary
    \begin{equation}
        \sel = \sum_{j=0}^{K-1} \op{j}{j} \otimes B_j + \sum_{j=K}^{2^b-1} \op{j}{j} \otimes \I_{a + n}.
    \end{equation}
    Then $(\prep_L^\dagger \otimes \I_{a + n}) \sel (\prep_R \otimes \I_{a + n})$ is an $(\alpha\beta, a + b, \alpha\varepsilon_1 + \beta\varepsilon_2)$-block encoding of $A$.\footnote{Note that a typo in \cite{gilyen2019quantum} mistakenly carried a factor of $\alpha$ in the $\beta\varepsilon_2$ error term.}
\end{proposition}

In this paper, we assume that we can produce state-preparation pairs with infinite precision. While some finite-precision machine error will be incurred in practice, this is vastly negligible to all other errors in our algorithm. Furthermore by the Solovay--Kitaev theorem, we only pay a polylogarithmic overhead for synthesizing the state-preparation pair (as well as any other quantum gates that we may need) from a finite gate set~\cite{dawson2006solovay}. If necessary, we trust the diligent reader to carry forth the machine precision throughout our exposition.

%% file: psuedo-choi_method.tex
\section{Parameter estimation from pseudo-Choi states}\label{sec:pseudo-choi_learning}

As the base learning algorithm, we will adapt the algorithm of \cite{castaneda2023hamiltonian}, which learns the Hamiltonian from a unique resource called its pseudo-Choi state. Given access to such a resource state, their algorithm can efficiently learn any Hamiltonian supported on a polynomial number of terms, regardless of any constraints such as locality that typically hinder other approaches. In this section, we review and streamline the key components of this algorithm. Readers familiar with these results may skip ahead to \cref{sec:structure_learning}.

In order to prepare the pseudo-Choi state from queries to the time evolution operator, \cite{castaneda2023hamiltonian} assumes the ability to perform controlled evolutions $\ctrl{}{e^{\pm\i H t}}$. Controlizing black-box unitaries is in general an impossible task~\cite{araujo2014quantum,gavorova2024topological}. However, because we have fractional-query access to $U(t)$ through the real parameter $t$, we can approximate $\ctrl{}{e^{-\i H t}}$ to arbitrary accuracy with no overhead in the evolution time~\cite{odake2024higher}. Instead, controlization only affects the time resolution and number of additional quantum gates;~we carry out this analysis in \cref{sec:controlization}. For simplicity of the present exposition, we will work with the exact controlled unitary.

\subsection{State preparation from time-evolution queries}

Castaneda and Wiebe~\cite{castaneda2023hamiltonian} defined the pseudo-Choi state with a reference control qubit as follows.

\begin{definition}
    Let $H$ be an $n$-qubit Hamiltonian. The \emph{pseudo-Choi state of $H$ with reference}\footnote{Later we will introduce a referenceless version, which is a $2n$-qubit state.} is the $(2n + 1)$-qubit state
    \begin{equation}
        \ket{\pchoiref(H)} \coloneqq \frac{(H \otimes \I_n) \ket{\Omega} \ket{0} + \ket{\Omega} \ket{1}}{\norm(H)},
    \end{equation}
    where $\ket{\Omega} = \frac{1}{\sqrt{2^n}} \sum_{b \in \{0, 1\}^n} \ket{b} \ket{b}$ and $\norm(H) = \sqrt{\frac{1}{2^n} \|H\|_F^2 + 1}$ is a normalization factor.
\end{definition}

The reference here will be vital for deducing the Hamiltonian parameters from measurements on the state. Observe that since $\|H\|_F \leq \sqrt{2^n} \|H\|$, it follows that $\norm(H) \leq \|H\| + 1$.

The actual pseudo-Choi state that we consider will be
\begin{equation}\label{eq:exact_PCS}
    \ket{\pchoiref(H/\Delta)} = \frac{(H/\Delta \otimes \I_n) \ket{\Omega} \ket{0} + \ket{\Omega} \ket{1}}{\norm(H/\Delta)},
\end{equation}
where $\Delta \geq 2\|H\|$ is a normalization on the Hamiltonian. This normalization is invoked because the pseudo-Choi states will be prepared from a block encoding of the Hamiltonian. Block encoded matrices must have norm at most $1$ to ensure unitarity of the encoding operator;~the further factor of $2$ ensures that the QSVT used to construct this block encoding is also properly normalized. \cite{low2017hamiltonian,gilyen2019quantum} show how to block encode $H$ from forward- and reverse-time queries to its time evolution operator. We opt to restate the formulation from \cite{gilyen2019quantum} here.

\begin{proposition}[{\cite[Corollary 71]{gilyen2019quantum}}]\label{prop:arcsin_log_unitary}
    Let $\varepsilon \in (0, \frac{\pi}{4}]$. Suppose that $U = e^{-\i H /\Delta}$ for some $\Delta \geq 2\|H\|$. There exists a $(\frac{\pi}{2}, 2, \varepsilon)$-BE of $H/\Delta$, constructed from $\O(\log(1/\varepsilon))$ queries to \emph{$\ctrl{}{U}$ and $\ctrl{}{U^\dagger}$}, and $\O(\log(1/\varepsilon))$ additional quantum gates.
\end{proposition}

\begin{proof}[Proof sketch.]
    We outline the proof idea for completeness and to clarify some details. We begin with the observation that
    \begin{equation}
     \ctrl{}{U}(Y \otimes \I_n)\ctrl{}{U^\dagger} = \begin{pmatrix}
         \sin(H/\Delta) & \cdot\ \\
         \cdot & \cdot\ 
     \end{pmatrix},
    \end{equation}
    where the matrix representation on the rhs is in the $\ket{\pm} = \frac{1}{\sqrt{2}}(\ket{0} \pm \ket{1})$ basis. This is a $(1, 1, 0)$-BE of $\sin(H/\Delta)$. Then, \cite[Lemma 70]{gilyen2019quantum} shows that there exists a polynomial $P(x)$ which approximates the $\arcsin(x)$ function on the domain $x \in [-\frac{1}{2}, \frac{1}{2}]$. Specifically, for any $\varepsilon' \in (0, \frac{1}{2}]$, this $P(x)$ is a polynomial of odd degree $d = \O(\log(1/\varepsilon'))$ obeying
    \begin{equation}
        \max_{-\frac{1}{2} \leq x \leq \frac{1}{2}} \l| P(x) - \frac{2}{\pi} \arcsin(x) \r| \leq \varepsilon'
    \end{equation}
    and $\max_{-\frac{1}{2} \leq x \leq \frac{1}{2}} |P(x)| \leq 1$. Hence, we can apply the QSVT on a single ancilla qubit to encode $P(\sin(H/\Delta))$. A polynomial transformation of degree $d$ requires $d$ queries to $U$ and $U^\dagger$ and $\O(d)$ additional quantum gates. This produces a $(\frac{\pi}{2}, 2, \frac{\pi}{2} \varepsilon')$-BE of $H/\Delta$;~letting $\varepsilon = \frac{\pi}{2} \varepsilon'$ furnishes the claim.
\end{proof}

From this block encoding, \cite{castaneda2023hamiltonian} shows that a pseudo-Choi state of an approximation to $H/\Delta$ can be prepared with probability greater than $\frac{1}{2}$.

\begin{proposition}[{\cite[Lemma 19]{castaneda2023hamiltonian}}]\label{prop:CW_pchoi_prep}
    Let $H$ be an $n$-qubit Hamiltonian and $\Delta \geq 2\|H\|$. Let $\varepsilon \in (0, \frac{\pi}{4}]$. With success probability at least $\frac{1}{2}$, we can prepare a copy of $\ket{\pchoiref(\frac{\widetilde{H}}{\Delta\pi/2})}$, where $\widetilde{H}$ is a Hermitian matrix obeying
    \begin{equation}
        \|H - \widetilde{H}\| \leq \varepsilon \Delta.
    \end{equation}
    This procedure uses $\O(\log(1/\varepsilon))$ queries to $U = e^{-\i H/\Delta}$, its inverse, and their controlled forms.
\end{proposition}


To produce $N$ copies of the state, $\O(N)$ queries to the block encoding suffice. This is a standard result derived from bounding the tail of a sequence of i.i.d.~Bernoulli trials. Following the argument given in \cite[Appendix B.6]{castaneda2023hamiltonian}, we prove the statement below in \cref{sec:bernoulli_bound} for completeness.

\begin{proposition}\label{prop:prob_state_prep}
    Let $Q \geq N \geq 1$ be integers. Let $\ket{\psi}$ be a state which is probabilistically prepared from a block encoding $B$, with success probability $p \in (0, 1]$. With high probability, we can prepare $N$ copies of $\ket{\psi}$ from
    \begin{equation}
        Q = \Theta(N/p)
    \end{equation}
    queries to $B$.
\end{proposition}

\subsection{Estimating with decoding operators}

Recall the Pauli decomposition of $H = \sum_{a=1}^m \lambda_a E_a$. In parallel, we define
\begin{equation}
    \widetilde{\lambda}_a = \frac{1}{2^n} \tr(E_a \widetilde{H}),
\end{equation}
which are the parameters in the block-encoded approximation. For each $a \in [m]$, \cite{castaneda2023hamiltonian} defines a set of so-called ``decoding operators'' on the $2n + 1$ qubits:
\begin{equation}\label{eq:decoding_ops}
\begin{split}
    O_a &= (E_a \otimes \I_n) \op{\Omega}{\Omega} \otimes \op{0}{1}, \quad a = 1, \ldots, m,\\
    O_\norm &= \op{\Omega}{\Omega} \otimes \op{1}{1}.
\end{split}
\end{equation}
These operators are defined such that their mean values with respect to $\ket{\pchoiref(\frac{\widetilde{H}}{\Delta\pi/2})}$ recover $\widetilde{\lambda}_a$, up to normalizations:\footnote{Note that the operators $O_a$ are not Hermitian, despite being constructed to have real-valued means. The postprocessing formulas used in classical shadows nonetheless easily accommodate complex-valued estimates.}
\begin{equation}
    \langle O_a \rangle = \frac{\widetilde{\lambda}_a}{\norm\l(\frac{\widetilde{H}}{\Delta\pi/2}\r{)^2} \Delta\pi/2}.
\end{equation}
While $\Delta$ is chosen by the learner, $\norm(\frac{\widetilde{H}}{\Delta\pi/2})$ is unknown. To access the value of this normalization factor, $O_\norm$ is used:
\begin{equation}
    \langle O_\norm \rangle = \norm\l(\frac{\widetilde{H}}{\Delta\pi/2}\r{)^{-2}}.
\end{equation}
Note that if we had not prepared the pseudo-Choi state with the reference control qubit, this normalization factor would have been challenging to infer.

The classical shadows protocol~\cite{huang2020predicting} over the Clifford group $\Cl(2^{2n+1})$ can be used to estimate all of these mean values with very efficient copy complexity. In brief, the shadow tomography protocol involves applying a random Clifford circuit to each copy of the state and then measuring in the standard basis. After collecting a sufficient number of samples, the measurement data is then classically postprocessed by a linear-inversion estimator which converges to the mean with strong theoretical guarantees. A major selling point of the technique is that, in terms of the explicit dependence on the number of observables, the copy complexity only scales logarithmically.

Let $\widehat{o}_j$ be the classical-shadows estimator for $\langle O_j \rangle$. \cite{castaneda2023hamiltonian} showed that their variances under Clifford shadow tomography are bounded by a constant:
\begin{equation}
    \max_{j \in [m] \cup \{\norm\}} \V[\widehat{o}_j] \leq 6.
\end{equation}
However, the estimate $\widehat{\lambda}_a$ of $\widetilde{\lambda}_a$ is a ratio of two point estimators,
\begin{equation}\label{eq:pchoi_shadow_estimator}
    \widehat{\lambda}_a \coloneqq \frac{\pi \Delta}{2} \frac{\Re[\widehat{o}_a]}{\Re[\widehat{o}_\norm]},
\end{equation}
which requires more than just bounding the variance to control its error. The analysis of \cite{castaneda2023hamiltonian} shows that learning each $\langle O_j \rangle$ up to precision $\varepsilon_{\mathrm{shadow}} = \Omega(\gamma/\Delta)$ suffices to achieve precision $\gamma$ in $\widehat{\lambda}_a$. Plugging this into the sample complexity for classical shadows~\cite{huang2020predicting}, they find that
\begin{equation}\label{eq:CW_sample_complexity}
    \O\l( \frac{\max_j \V[\widehat{o}_j] \log(m/\delta)}{\varepsilon_{\mathrm{shadow}}^2} \r) = \O\l( \frac{\Delta^2 \log(m/\delta)}{\gamma^2} \r)
\end{equation}
copies of the pseudo-Choi state suffice to produce an estimate $\widehat{\lambda} \in \R^m$ such that
\begin{equation}
    \Pr(\|\widetilde{\lambda} - \widehat{\lambda}\|_\infty \leq \gamma) \geq 1 - \delta.
\end{equation}
Note that the original analysis~\cite[Proposition 15]{castaneda2023hamiltonian} uses a propagation of uncertainty argument under a linear approximation;~a more rigorous analysis by the Taylor expansion of the random variable $\widehat{o}_a/\widehat{o}_\norm$ leads to the same conclusion (up to constant factors).

Finally, one must translate this estimate of $\widetilde{H}$ to an estimate of $H$. By H\"{o}lder's inequality, each parameter has systematic error at most
\begin{equation}\label{eq:block_error_coeff}
\begin{split}
    |\lambda_a - \widetilde{\lambda}_a| &= \frac{1}{2^n} |{\tr(E_a (H - \widetilde{H}))}|\\
    &\leq \frac{1}{2^n} \|E_a\|_1 \|H - \widetilde{H}\|\\
    &= \|H - \widetilde{H}\|.
\end{split}
\end{equation}
Recall from \cref{prop:CW_pchoi_prep} that $\|H - \widetilde{H}\| \leq \varepsilon\Delta$. Meanwhile, the statistical error is controlled by the shadow tomography analysis, which guarantees $\|\widetilde{\lambda} - \widehat{\lambda}\|_\infty \leq \gamma$ with a sufficient number of measurements. Altogether, by a triangle inequality the total error is
\begin{equation}\label{eq:CW_pchoi_total_error}
\begin{split}
    \|\lambda - \widehat{\lambda}\|_\infty &\leq \|\lambda - \widetilde{\lambda}\|_\infty + \|\widetilde{\lambda} - \widehat{\lambda}\|_\infty\\
    &\leq \O(\varepsilon\Delta) + \gamma.
\end{split}
\end{equation}

To bound this by $\epsilon$, we take $\O(\log(\Delta/\epsilon))$ queries to the controlled time evolution to generate a sufficiently accurate copy of the pseudo-Choi state. Because the probability of successful preparation is $\geq \frac{1}{2}$, only a constant number of repetitions are required to generate a copy from the block encoding (\cref{prop:prob_state_prep}). This bounds the first term of \cref{eq:CW_pchoi_total_error} by, say, $\epsilon/2$. Consuming $\Ot(\Delta^2/\epsilon^2)$ such copies via Clifford shadows bounds the second term, also by $\epsilon/2$. The total number of queries is therefore $\Ot(\Delta^2/\epsilon^2)$, and since each query runs for time $1/\Delta$, the total evolution time is $\Ot(\Delta/\epsilon^2)$. This is one of the main results of \cite{castaneda2023hamiltonian}.

\begin{proposition}[{\cite[Theorem 26, modified for $\ell_\infty$-error]{castaneda2023hamiltonian}}]\label{prop:CW_learning_algorithm}
    Let $H = \sum_{a=1}^m \lambda_a E_a$ be an unknown Hamiltonian and let $\Delta \geq 2 \|H\|$. Assuming access to $U(t) = e^{-\i H t}$, $U(t)^\dagger = U(-t)$, and their multi-controlled forms, there exists an algorithm which can learn each $\lambda_a$ to additive error $\epsilon$. The total number of queries to the time evolution operator is $\Ot(\Delta^2/\epsilon^2)$, the total evolution time is $\Ot(\Delta/\epsilon^2)$, and the minimum time resolution is $\Omega(1/\Delta)$. The amount of classical and quantum computation scales at most polynomially in all parameters.
\end{proposition}

\begin{remark}
    This algorithm can perform structure learning of arbitrary Hamiltonians, with some important caveats. First, suppose the pool of possible Hamiltonian terms is restricted to some polynomially sized set. In that case, shadow tomography gives sufficient samples to estimate all the possible terms, and the structure can be learned by selecting only those which have large magnitude. For example, the set of $k$-local terms has size $\O(n^k)$, and this is the setting originally considered in \cite{castaneda2023hamiltonian}. Second, if one allows an exponential amount of classical computation, then the structure can be learned by brute-force processing the shadows over all $4^n - 1$ Pauli operators. Although computationally inefficient, the query complexity is only increased by a factor of $\O(\log m) = \O(\log 4^n) = \O(n)$ for the shadow tomography step. \cite{caro2024learning} uses a similar idea that is also computationally inefficient in this setting for the same reasons.
\end{remark}

%% file: structure_learning.tex
\section{Identifying Pauli spectra with pseudo-Choi states}\label{sec:structure_learning}

In this section, we show how to identify the structure $\{E_a : |\lambda_a| \geq \epsilon\}$ of $H = \sum_{a=1}^m \lambda_a E_a$, using a total evolution time of $\Ot(\Delta/\epsilon^2)$. In the literature, this problem is typically encountered as learning where the Pauli spectrum of a \emph{unitary} operator has large mass~\cite{montanaro2010quantum,angrisani2023learning}. Here, we instead want to identify the Pauli spectrum of a non-unitary operator, $H$. The pseudo-Choi state provides us a convenient resource to accomplish this task.

To deduce the identity of the unknown terms, the key idea is to sample from the referenceless pseudo-Choi state of $H$. This state can be prepared by essentially the same techniques as for the pseudo-Choi state with reference, described in \cref{sec:pseudo-choi_learning}. The only changes are that we do not need to control the block encoding, and that we have different guarantees on the success probability.

First, let us define these states.

\begin{definition}
    Let $H$ be an $n$-qubit Hamiltonian. The \emph{referenceless pseudo-Choi state of $H$} is the $2n$-qubit state
    \begin{equation}
        \ket{\pchoi(H)} \coloneqq \frac{(H \otimes \I_n) \ket{\Omega}}{\|(H \otimes \I_n) \ket{\Omega}\|},
    \end{equation}
    where $\ket{\Omega} = \frac{1}{\sqrt{2^n}} \sum_{b \in \{0, 1\}^n} \ket{b} \ket{b}$ and $\|(H \otimes \I_n) \ket{\Omega}\| = \frac{1}{\sqrt{2^n}} \|H\|_F$.
\end{definition}

Unlike the pseudo-Choi state with reference, this referenceless variant is invariant under rescaling $H \mapsto sH$ for any $s \in \C\setminus\{0\}$. Let us briefly outline a method for preparing this state if one has access to both positive- and negative-time evolutions. Without negative-time dynamics, a different approach must be taken to block encode the Hamiltonian. This results in an entirely different query complexity;~see \cref{sec:time_reversal_free} for that setting.

\begin{lemma}\label{lem:refless_pchoi_prep}
    Let $H = \sum_{a=1}^m \lambda_a E_a$ be an $n$-qubit Hamiltonian, $\Delta \geq 2 \|H\|$, and $\varepsilon > 0$. With success probability $\Theta(\|\widetilde{\lambda}/\Delta\|^2)$, we can prepare a copy of $\ket{\pchoi(\widetilde{H})}$, where $\widetilde{H}$ is a Hermitian matrix obeying
    \begin{equation}
        \|H - \widetilde{H}\| \leq \varepsilon
    \end{equation}
    and $\widetilde{\lambda}_a = \frac{1}{2^n} \tr(E_a \widetilde{H})$. This procedure uses $\O(\log(\Delta/\varepsilon))$ queries to $U = e^{-\i H/\Delta}$, its inverse, and their controlled forms.
\end{lemma}

\begin{proof}
    The preparation of this state is similar to that of the variant with reference, see \cref{prop:CW_pchoi_prep}. The main difference here is that we do not need a controlled form of the block encoding. Let $B$ be the $(\frac{\pi}{2}, 2, 0)$-BE of $\widetilde{H}/\Delta$ produced from $\O(\log(\Delta/\varepsilon))$ queries, where $\|H - \widetilde{H}\| \leq \varepsilon$. Applying $B \otimes \I_n$ to $\ket{00}\ket{\Omega}$ and measuring the first two qubits, if we obtain $\ket{00}$ in that register then we have successfully prepared $\ket{\pchoi(\widetilde{H})}$ in the remaining $2n$ qubits. This occurs with probability
    \begin{equation}
    \begin{split}
        \Pr(00) &= \l\| \l(\frac{\widetilde{H}}{\Delta \pi/2} \otimes \I_n\r) \ket{\Omega} \r{\|^2}\\
        &= \l( \frac{2}{\pi} \r{)^2} \frac{\|\widetilde{H}/\Delta\|^2_F}{2^n}\\
        &= \l( \frac{2}{\pi} \r{)^2} \| \widetilde{\lambda}/\Delta \|^2.
    \end{split}
    \end{equation}
\end{proof}

Let us remark briefly about the applicability of this result. In the context of the bootstrap framework, we will apply \cref{lem:refless_pchoi_prep} in the first step where $\widehat{H}_0 = 0$. Once we have a nontrivial estimate $\widehat{H}_j \neq 0$, we synthesize the residual Hamiltonian $H - \widehat{H}_j$ and learn from that instead. In that case, the proof argument is the same but the constant $(2/\pi)^2$ becomes $1/4$.

Returning to the present setting, note that the block-encoded approximation $\widetilde{H} = \sum_{a=1}^{m'} \widetilde{\lambda}_a E_a$ may have $(m' - m)$ potentially superfluous error terms. For some $\gamma \in (0, 1)$, define the set
\begin{equation}
    S_\gamma(\widetilde{H}) \coloneqq \{ E_a : |\widetilde{\lambda}_a| > \gamma \}.
\end{equation}
This labels the terms in $\widetilde{H}$ that have weight more than $\gamma$. It suffices only to learn the terms in $S_\gamma(\widetilde{H})$ because, if $E_a \notin S_\gamma(\widetilde{H})$, then $\widehat{\lambda}_a = 0$ is a $\gamma$-accurate estimate of $\widetilde{\lambda}_a$. For small enough $\gamma$, such terms can safely be ignored.

In order to learn this set, we will appeal to the technique of Bell sampling. In our case, we sample from the pseudo-Choi state (instead of two copies of an $n$-qubit state, as was introduced in \cite{montanaro2017learning}).

\begin{lemma}
    Bell measurements across the bipartition of the $2n$-qubit state $\ket{\pchoi(\widetilde{H})}$ sample from the distribution
    \begin{equation}
        \mathcal{D} = \l\{\Pr(E_a) = \frac{|\widetilde{\lambda}_a|^2}{\ell^2} : a \in [m'] \r\},
    \end{equation}
    where $\ell = \|\widetilde{\lambda}\|$.
\end{lemma}

\begin{proof}
    Bell measurements, implemented by a depth-$2$ Clifford circuit, sample from the basis $\{ \ket{P} = (P \otimes \I_n) \ket{\Omega} : P \in \Pauli{n} \}$. On the state $\ket{\pchoi(\widetilde{H})}$, this induces the probability distribution
    \begin{equation}
    \begin{split}
        |\ip{P}{\pchoi(\widetilde{H})}|^2 &= \frac{|\ev{(P^\dagger \otimes \widetilde{H})}{\Omega}|^2}{\|(\widetilde{H} \otimes \I_n) \ket{\Omega}\|^2}\\
        &= \frac{| \frac{1}{2^n} \tr(P^\dagger \widetilde{H}) |^2}{\ell^2}.
    \end{split}
    \end{equation}
    We have $| \frac{1}{2^n} \tr(P^\dagger \widetilde{H}) |^2 = |\widetilde{\lambda}_a|^2$ if $P = E_a$, and $0$ otherwise.
\end{proof}

Our goal then is to draw samples from $\mathcal{D}$ until we see all elements of $S_\gamma(\widetilde{H})$. The following lemma bounds the number of Bell measurements sufficient to accomplish this task with high probability.

\begin{lemma}\label{lem:coupon_collector}
    Let $\gamma, \delta \in (0, 1)$. Let $\widetilde{H}$ be a Hamiltonian and suppose we have access to copies of its referenceless pseudo-Choi state. With probability at least $1 - \delta$, taking $\O(\|\widetilde{\lambda}\|^2 \log(|S_\gamma(\widetilde{H})|/\delta) / \gamma^2)$ Bell measurements of $\ket{\pchoi(\widetilde{H})}$ suffices to observe all elements of $S_\gamma(\widetilde{H})$.
\end{lemma}

\begin{proof}
    This problem is a variant of the classic coupon collector problem:~given a distribution $\mathcal{D}$ over $m'$ ``coupons,'' how many draws does it take until we have collected at least one of each coupon? Our bound follows from a modification of \cite[Lemma 12]{shioda2007some}. In particular, we relax the task by not aiming to collect the rare coupons (i.e., those outside of $S_\gamma(\widetilde{H})$).
    
    Let $X_\gamma$ be the random variable defined as the number of draws until all coupons from $S_\gamma(\widetilde{H})$ are obtained. For each integer $N > 0$, define $A_a(N)$ as the event that, after drawing $N$ coupons, none of them are $E_a$. The total probability that after $N$ draws, we have yet to obtain all coupons (from $S_\gamma(\widetilde{H})$), can be estimated by a union bound:
    \begin{equation}
        \Pr(X_\gamma > N) = \Pr\l(\bigcup_{a \in S_\gamma(\widetilde{H})} A_a(N) \r) \leq \sum_{a \in S_\gamma(\widetilde{H})} \Pr(A_a(N)) = \sum_{a \in S_\gamma(\widetilde{H})} [1 - \Pr(E_a)]^N.
    \end{equation}
    Let $\ell = \|\widetilde{\lambda}\|^2$. For each $E_a \in S_\gamma(\widetilde{H})$, we have that $\Pr(E_a) = |\widetilde{\lambda}_a|^2/\ell^2 > \gamma^2/\ell^2$. The inequality $1 - x \leq e^{-x}$ then gets
    \begin{equation}
        \Pr(X_\gamma > N) < |S_\gamma(\widetilde{H})| \exp\l( -\frac{N \gamma^2}{\ell^2} \r).
    \end{equation}
    Setting this bound as the maximum failure probability $\delta$ and solving for $N$ proves the claim.
\end{proof}

\begin{remark}
    The Bell sampling algorithm can also discover elements $E_{a'} \notin S_\gamma(\widetilde{H})$. At this current stage in the algorithm, we do not have enough information to determine inclusion or exclusion, so we store all elements observed. As long as we have at least a superset of $S_\gamma(\widetilde{H})$, we say that the structure learning algorithm has succeeded. In the next stage, shadow tomography can easily estimate an overcomplete set of terms with only a logarithmic sampling overhead, per \cref{eq:CW_sample_complexity}. Any terms that are too small are then discarded;~more precisely, if we set the shadow tomography error to $\epsilon/2$, then discarding all $|\widehat{\lambda}_a| \leq \epsilon$ ensures property 2 of \cref{prob:1}. And because we only ever make a polynomial number of queries, we never identify more than a polynomial number of candidate terms.
\end{remark}

The last piece we need is a conversion from learning $S_{\gamma}(\widetilde{H})$ to learning $S_{\epsilon}(H)$, given that $\|H - \widetilde{H}\| \leq \varepsilon$. 

\begin{lemma}\label{lem:big_set_conversion}
    Let $\varepsilon, \epsilon \in (0, 1)$. Let $H, \widetilde{H}$ be two Hamiltonians such that
    \begin{align}
        H &= \sum_{a=1}^m \lambda_a E_a,\\
        \widetilde{H} &= \sum_{a=1}^m \widetilde{\lambda}_a E_a + \sum_{a'=m+1}^{m'} \widetilde{\lambda}_{a'} E_{a'},
    \end{align}
    where $m' \geq m$. Suppose $\|H - \widetilde{H}\| \leq \varepsilon$. Then $S_\epsilon(H) \subseteq S_{\epsilon + \varepsilon}(\widetilde{H})$, and furthermore, $E_{a'} \notin S_{\epsilon + \varepsilon}(\widetilde{H})$ for all $a' = m + 1, m + 2, \ldots, m'$.
\end{lemma}

\begin{proof}
    For all $E_a \in S_{\epsilon + \varepsilon}(\widetilde{H})$ we have $|\widetilde{\lambda}_a| > \epsilon + \varepsilon$. By a triangle inequality and H\"{o}lder's inequality (\cref{eq:block_error_coeff}), it holds that $|\widetilde{\lambda}_a| \leq |\lambda_a| + \varepsilon$. Therefore all such $a$ also correspond to $|\lambda_a| > \epsilon$, which characterizes the set $S_\epsilon(H)$. This demonstrates the first claim.
    
    For the second claim, consider any $a' \in \{m + 1, m + 2, \ldots, m'\}$. Because $H$ has no support on such terms, $\lambda_{a'} = 0$ and so $|\widetilde{\lambda}_{a'}| \leq \varepsilon$. But since $\epsilon > 0$, this magnitude is strictly less than $\varepsilon + \epsilon$. Hence all such $\widetilde{\lambda}_{a'}$ are too small to lie in $S_{\epsilon + \varepsilon}(\widetilde{H})$.
\end{proof}

Combining \cref{lem:refless_pchoi_prep,lem:coupon_collector,lem:big_set_conversion}, we arrive at the main result of this section.

\begin{theorem}\label{thm:structure_learning_base}
    Let $H = \sum_{a=1}^m \lambda_a E_a$ be an unknown $n$-qubit Hamiltonian, $\Delta \geq 2 \|H\|$, and $\epsilon > 0$. Define the set $S_\epsilon(H) = \{E_a : |\lambda_a| > \epsilon\}$. There exists an algorithm which learns, with high probability, at least all the elements of $S_\epsilon(H)$. The algorithm makes
    \begin{equation}
        Q = \O\l( \frac{\Delta^2 \log(m) \log(\Delta/\epsilon)}{\epsilon^2} \r).
    \end{equation}
    queries to \emph{$\ctrl{}{e^{\pm \i H/\Delta}}$} and uses $\poly(n, m, 1/\epsilon)$ quantum and classical computation.
\end{theorem}

\begin{proof}
    \cref{lem:coupon_collector} tell us that, with high probability,
    \begin{equation}
        N = \O\l(\frac{\|\widetilde{\lambda}\|^2 \log |S_\gamma(\widetilde{H})|}{\gamma^2}\r)
    \end{equation}
    copies of $\ket{\pchoi(\widetilde{H})}$ suffice to learn $S_\gamma(\widetilde{H})$ via Bell measurements. Meanwhile, \cref{lem:refless_pchoi_prep} tells us that the probability of successfully generating a single copy is
    \begin{equation}
        p = \Theta\l(\frac{\|\widetilde{\lambda}\|^2}{\Delta^2}\r).
    \end{equation}
    By \cref{prop:prob_state_prep}, we need to take $Q' = \O(N/p)$ trials to produce the $N$ copies. Therefore from
    \begin{equation}
        Q' = \O\l(\frac{\Delta^2 \log |S_\gamma(\widetilde{H})|}{\gamma^2}\r)
    \end{equation}
    queries to the block encoding of $\widetilde{H}/\Delta$, or equivalently $Q = \O(Q' \log(\Delta/\varepsilon))$ queries to the time evolution operator, we learn the set $S_\gamma(\widetilde{H})$ with high probability. Finally by \cref{lem:big_set_conversion}, setting $\gamma = \epsilon + \varepsilon$ ensures that the set $S_{\epsilon + \varepsilon}(\widetilde{H})$ that we have learned contains $S_\epsilon(H)$ as a subset. \cref{lem:big_set_conversion} also assures us that $|S_{\epsilon + \varepsilon}(\widetilde{H})| \leq m$ since it cannot contain any superfluous terms. Taking $\varepsilon = \Theta(\epsilon)$ sufficiently small concludes the proof.
\end{proof}

Comparing this result to \cref{prop:CW_learning_algorithm}, we find that the query and evolution-time complexity of both parameter learning and structure learning from pseudo-Choi states are identical. Therefore even without invoking the precision bootstrap, we have demonstrated an algorithm that can perform structure learning on a completely arbitrary $n$-qubit Hamiltonian, using $\Ot(\|H\|/\epsilon^2)$ total evolution time.

To conclude this section, we make a remark regarding the identity component in $H$.

\begin{remark}[On traceful Hamiltonians]
    If $H$ is not traceless, then the identity coefficient (which we are typically not interested in learning) may wash out the distribution. Asymptotically, this only slows down the learning procedure by a constant factor. However, we have a nicer guarantee using the approximate controlization described in \cref{rem:ctrl_traceless}. Those techniques, which we use to prepare the psuedo-Choi states, will automatically extract the traceless part of $H$. Hence moving forward we may suppose $\tr H = 0$ without loss of generality.
\end{remark}

%% file: residual_hamiltonian.tex
\section{Learning with the residual Hamiltonian}\label{sec:residual}

Let $\widehat{H} = \sum_a \widehat{\lambda}_a E_a$ be the current best estimate of $H$, satisfying $\|\lambda - \widehat{\lambda}\|_\infty \leq \eta$. The key to the precision bootstrap is the ability to learn from $R/\eta$, where
\begin{equation}
    R \coloneqq H - \widehat{H}
\end{equation}
is the ``residual'' Hamiltonian. By scaling up the residual by a factor of $1/\eta$, it suffices to learn $R/\eta$ with merely constant precision $\epsilon_0 \leq \frac{1}{2}$. In this section, we will show how to block encode $R/\eta$ and prepare its pseudo-Choi states. This will give us the majority of the proof of \cref{thm:tr_learning}.

\subsection{Block encoding by spectral amplification}

Recall that \cref{prop:arcsin_log_unitary} gave us a $(\frac{\pi}{2}, 2, \varepsilon)$-BE of $H/\Delta$ from queries to its time evolution, provided that $\Delta \geq 2\|H\|$. To make a concrete choice, since $\|H\| \leq m$ we shall fix
\begin{equation}
    \Delta = 2m
\end{equation}
here. This will guarantee that all the operators we aim to block encode have are sufficiently normalized.

Next, since our description of $\widehat{H}/\Delta$ is already an LCU, we can synthesize its block encoding with standard techniques (\cref{prop:LCBE}). However, a technicality we need to address is to match normalizations between the two block encodings before we take their difference.\footnote{Alternatively we could make the appropriate choice of nonuniform coefficients in the linear combination;~we opt for the approach as presented to simplify the exposition.}

\begin{lemma}\label{lem:rescaled_est_H_norm}
    Let $\widehat{H} = \sum_{a=1}^m \widehat{\lambda}_a E_a$ and $\Delta = 2m$. There exists a $(\frac{\pi}{2}, \lceil \log_2 m \rceil + 1, 0)$-BE of $\widehat{H}/\Delta$.
\end{lemma}

\begin{proof}
    Let $\widehat{B}'$ be the $(\|\widehat{\lambda}\|_1, \lceil \log_2 m \rceil, 0)$-BE of $\widehat{H}$ as given by \cref{prop:LCBE}. We aim to convert the subnormalization on the block-encoded matrix according to
    \begin{equation}
        \widehat{B}' = \begin{pmatrix}
            \frac{\widehat{H}}{\|\widehat{\lambda}\|_1} & \cdot\ \\
            \cdot & \cdot\ 
        \end{pmatrix} \mapsto
        \begin{pmatrix}
            \frac{\widehat{H}}{\Delta \pi/2} & \cdot\ \\
            \cdot & \cdot\ 
        \end{pmatrix}.
    \end{equation}
    Define single-qubit the rotation
    \begin{equation}
        R_y(\theta) = \begin{pmatrix}
            \cos\l(\frac{\theta}{2}\r) & -\sin\l(\frac{\theta}{2}\r)\\
            \sin\l(\frac{\theta}{2}\r) & \cos\l(\frac{\theta}{2}\r)
        \end{pmatrix}.
    \end{equation}
    Set $\theta = 2 \arccos\l( \frac{2\|\widehat{\lambda}\|_1}{\pi \Delta} \r)$, which is a real-valued angle since $0 \leq \frac{2\|\widehat{\lambda}\|_1}{\pi \Delta} \leq \frac{1}{\pi}$. Then
    \begin{equation}
        \widehat{B} \coloneqq (R_y(\theta) \otimes \I_{\lceil \log_2 m \rceil}) \cdot \ctrl{0}{\widehat{B}'} = \begin{pmatrix}
            \frac{2\|\widehat{\lambda}\|_1}{\pi \Delta} \widehat{B}' & \cdot\ \\
            \cdot & \cdot\ 
        \end{pmatrix}
    \end{equation}
    has the desired subnormalization in the top-left block, using one more ancilla qubit.
\end{proof}

The block encoding for the residual Hamiltonian immediately follows.

\begin{lemma}\label{lem:unamp_residual}
    Define $\widetilde{R} \coloneqq \widetilde{H} - \widehat{H}$ where $\frac{1}{\Delta} \|H - \widetilde{H}\| \leq \varepsilon$. There exists a $(\pi, \lceil \log_2 m \rceil + 2, 0)$-BE of $\widetilde{R}/\Delta$, using $\O(\log(1/\varepsilon))$ queries to \emph{$\ctrl{}{e^{\pm \i H/\Delta}}$}.
\end{lemma}

\begin{proof}
    Here, let us interpret $\widetilde{B}$ as a $(\frac{\pi}{2}, 2, 0)$-BE of $\widetilde{H}/\Delta$ given by \cref{prop:arcsin_log_unitary}. Let $\widehat{B}$ be the $(\frac{\pi}{2}, \lceil \log_2 m \rceil + 1, 0)$-BE of $\widehat{H}/\Delta$ given by \cref{lem:rescaled_est_H_norm}, and let $(\prep_L, \prep_R)$ be a $(2, 1, 0)$-state-preparation pair for $y = (1, -1)$. Then by \cref{prop:LCBE}, we can produce a $(\pi, \lceil \log_2 m \rceil + 2, 0)$-BE for $\widetilde{R}/\Delta$ using a single query to $\ctrl{}{\widetilde{B}}$ and $\ctrl{}{\widehat{B}}$ each. In terms of the query complexity, a single call to $\ctrl{}{\widetilde{B}}$ costs $\O(\log(1/\varepsilon))$ queries to $e^{-\i H/\Delta}$, as stated in \cref{prop:arcsin_log_unitary}.
\end{proof}

Observe that this synthesized residual matrix has small operator norm, because
\begin{equation}\label{eq:norm_of_raw_residual}
\begin{split}
    \|\widetilde{R}/\Delta\| &= \frac{1}{\Delta} \|\widetilde{H} - \widehat{H}\|\\
    &\leq \frac{1}{\Delta} \l( \|\widetilde{H} - H\| + \|H - \widehat{H}\| \r)\\
    &\leq \varepsilon + \sum_{a=1}^m \frac{|\lambda_a - \widehat{\lambda}_a|}{2m}\\
    &\leq \varepsilon + \frac{\eta}{2}.
\end{split}
\end{equation}
We need to boost this norm to $\O(1)$ in order to learn with sufficient precision. This can be accomplished by uniform spectral amplification~\cite{low2017hamiltonian}, a matrix-generalization of amplitude amplification. Note that a further generalization called uniform singular value amplification exists for applications beyond Hermitian matrices~\cite{gilyen2019quantum}, however we will not need the strength of that extension here.

\begin{proposition}[{\cite[Theorem 2]{low2017hamiltonian}}]\label{prop:USA}
    Let $B$ be an $(\alpha, q, 0)$-BE for a Hermitian matrix $H$. Choose some $\beta \in [\|H\|, \alpha]$. For any $0 < \varepsilon \leq \O(\beta/\alpha)$, we can construct a $(2\beta, q + 2, 2\beta\varepsilon)$-BE of $H$, using $N = \O(\frac{\alpha}{\beta} \log(1/\varepsilon))$ queries to \emph{$\ctrl{}{B}$} and $\O(Nq)$ additional quantum gates.
\end{proposition}

With this technique, we can block encode the amplified residual Hamiltonian.

\begin{theorem}\label{thm:amp_BE_residual}
    Let $H = \sum_{a=1}^m \lambda_a E_a$ be an unknown $n$-qubit Hamiltonian and $\widehat{H} = \sum_{a=1}^m \widehat{\lambda}_a E_a$ an estimate such that $\| \lambda - \widehat{\lambda} \|_\infty \leq \eta$ for some $\eta > 0$. Let $\Delta = 2m$, $q = \lceil \log_2 m \rceil + 2$, and fix some error parameters $0 < \varepsilon, \varepsilon_{\mathrm{amp}} \leq \O(\eta)$. There exists a $(2, q + 2, \varepsilon + \varepsilon_{\mathrm{amp}})$-BE of
    \begin{equation}
        \frac{R}{\eta\Delta} \equiv \frac{H - \widehat{H}}{\eta\Delta}
    \end{equation}
    which costs $\O(\log(1/(\varepsilon \eta))\log(1/\varepsilon_{\mathrm{amp}})/\eta)$ queries to multi-controlled variants of $e^{\pm \i H/\Delta}$.
\end{theorem}

\begin{proof}
Let
\begin{equation}
    T \coloneqq \frac{\widetilde{R}}{2\eta \Delta},
\end{equation}
be the ``target'' operator, which is what we want to amplify up to. Note that by \cref{eq:norm_of_raw_residual} we have
\begin{equation}
    \|T\| \leq \frac{\varepsilon + \eta/2}{2\eta},
\end{equation}
which is at most $\frac{1}{2}$ as long as $\varepsilon \leq \eta/2$. Let $B$ be the $(\pi, q, 0)$-BE of $\widetilde{R}/\Delta$, given by \cref{lem:unamp_residual}. Equivalently, we interpret this as a $(\frac{\pi}{2\eta}, q, 0)$-BE for $T$. Choose $\beta = \frac{1}{2}$, which is guaranteed to lie within the interval $[\|T\|, \frac{\pi}{2\eta}]$. Also note that $\varepsilon_{\mathrm{amp}} \leq \O(\eta)$. Then, we apply uniform spectral amplification to produce $B_{\mathrm{amp}}$, which by \cref{prop:USA} is a $(1, q + 2, \varepsilon_{\mathrm{amp}})$-BE of $T$. The circuit for this transformation costs $\O(\frac{\pi}{\eta} \log(1/\varepsilon_{\mathrm{amp}}))$ queries to $\ctrl{}{B}$, or equivalently
\begin{equation}\label{eq:queries_for_amp_res}
    Q = \O\l(\frac{\log(1/\varepsilon) \log(1/\varepsilon_{\mathrm{amp}})}{\eta} \r)
\end{equation}
queries to $\ctrl{}{e^{\pm \i H/\Delta}}$. To get the block-encoding parameters stated in the theorem, observe that $2T = \frac{\widetilde{R}}{\eta\Delta}$ is $2\varepsilon_{\mathrm{amp}}$-close to the amplified matrix. Meanwhile, $\frac{1}{\eta\Delta} \|R - \widetilde{R}\| \leq \varepsilon/\eta$. Rescaling error parameters $\varepsilon_{\mathrm{amp}} \leftarrow \varepsilon_{\mathrm{amp}}/2$ and $\varepsilon \leftarrow \varepsilon\eta$ concludes the proof.
\end{proof}

\subsection{Parameter estimation}

Now we consider the error parameters sufficient for learning the residual Hamiltonian, within the context of the bootstrap algorithm. Recall that we only need to learn $R/\eta$ up to constant error of $\epsilon_0 = \frac{1}{2}$.

Let $T'$ be the matrix encoded by $B_{\mathrm{amp}}$. By \cref{thm:amp_BE_residual}, it has error
\begin{equation}
    \l\| \frac{R}{\eta\Delta} - 2T' \r\| \leq \varepsilon + \varepsilon_{\mathrm{amp}}.
\end{equation}
Define
\begin{align}
    r_a &\coloneqq \frac{1}{2^n} \tr[E_a (R/\eta)],\\
    r_a' &\coloneqq \frac{1}{2^n} \tr[ E_a (2\Delta T') ].
\end{align}
By H\"{o}lder's inequality, we get
\begin{equation}\label{eq:BE_res_coeff_err}
\begin{split}
    |r_a - r_a'| &\leq \l\| \frac{R}{\eta} - 2\Delta T' \r\|\\
    &\leq \Delta (\varepsilon + \varepsilon_{\mathrm{amp}}).
\end{split}
\end{equation}
We want to learn $r_a$ to within constant additive error $\epsilon_0 = \frac{1}{2}$, so it suffices to bound \cref{eq:BE_res_coeff_err} by a small constant as well. Choosing $\varepsilon = \varepsilon_{\mathrm{amp}} = c \epsilon_0/\Delta$ for some sufficiently small constant $c \ll \frac{1}{2}$, we get $|r_a - r_a'| \leq 2c \epsilon_0$. By \cref{thm:amp_BE_residual}, the query complexity for producing the block encoding of such a $T'$ is
\begin{equation}
    Q = \O\l(\frac{\log(\Delta/\eta) \log(\Delta)}{\eta} \r),
\end{equation}
where we have suppressed the constants $c, \epsilon_0$. Then given copies of $\ket{\pchoiref(T')}$, we can run the shadow tomography protocol of \cite{castaneda2023hamiltonian} to learn an estimate $\widehat{r}'_a$ of $r'_a$. By \cref{eq:CW_sample_complexity}, with high probability we get an $\epsilon_0$-accurate estimate by measuring
\begin{equation}
    N = \O\l( \frac{\Delta^2 \log m}{\epsilon_0^2} \r).
\end{equation}
copies of $\ket{\pchoiref(T')}$. Note that, using the decoding operators $O_a$ and $O_\norm$ as defined in \cref{eq:decoding_ops}, our estimate using $\ket{\pchoiref(T')}$ has slightly modified constants in its definition:
\begin{equation}\label{eq:residual_parameter_shadow_estimate}
    \widehat{r}_a' \coloneqq 2\Delta \frac{\Re[\widehat{o}_a]}{\Re[\widehat{o}_\norm]}.
\end{equation}

Now we consider the preparation of $\ket{\pchoiref(T')}$. Conditioned on its successful preparation from a query to $\ctrl{0}{B_{\mathrm{amp}}}$, the total number of queries required is $NQ = \Ot(\Delta^2/\eta)$. Meanwhile, parallel to the claim in \cref{prop:CW_pchoi_prep}, the success probability is above $\frac{1}{2}$, since
\begin{equation}
\begin{split}
    \Pr(\text{Preparing } \ket{\pchoiref(T')}) &= \frac{1}{2} \l\| ( T' \otimes \I_n ) \ket{\Omega} \ket{0} + \ket{\Omega} \ket{1} \r{\|^2}\\
    &= \frac{1}{2} \l( \frac{1}{2^n} \|T'\|_F^2 + 1 \r)\\
    &\geq \frac{1}{2}.
\end{split}
\end{equation}
Therefore by \cref{prop:prob_state_prep}, only a constant number of queries suffices to prepare a copy of the pseudo-Choi state with reference, for a total of $\O(NQ)$ queries and $\O(NQ/\Delta)$ evolution time. The theorem below follows.

\begin{theorem}\label{thm:amp_residual_learning}
    Let $H = \sum_{a=1}^m \lambda_a E_a$ be an unknown $n$-qubit Hamiltonian and $\widehat{H} = \sum_{a=1}^m \widehat{\lambda}_a E_a$ an estimate such that $\|\lambda - \widehat{\lambda}\|_\infty \leq \eta$ for some $\eta \in (0, 1]$. There exists a learning algorithm which, with high probability, produces an estimate $\widehat{r}_a'$ of $(\lambda_a - \widehat{\lambda}_a)/\eta$ up to constant additive error (say, $\epsilon_0 = \frac{1}{2}$) for each $a \in [m]$. The algorithm makes $\Ot(m^2/\eta)$ queries to \emph{$\ctrl{}{e^{\pm\i H/(2m)}}$} and uses $\poly(n, m, 1/\eta)$ quantum and classical computation.
\end{theorem}

\subsection{Structure learning}

To perform structure learning with Heisenberg-limited scaling (up to the $\log(1/\epsilon)$ factor), the algorithm of \cref{sec:structure_learning} must also be placed into the bootstrap framework. We describe how that procedure works in this subsection.

Suppose we have copies of the referenceless pseudo-Choi state $\ket{\pchoi(T')}$, prepared from queries to $B_{\mathrm{amp}}$. (We analyze the complexity of their preparation later.) We run the algorithm of \cref{lem:coupon_collector} to acquire a set containing $S_\gamma(2\Delta T') = \{E_a : |r_a'| > \gamma\}$;~this consumes
\begin{equation}
    N = \Ot\l( \frac{\|r'\|^2}{\gamma^2} \r)
\end{equation}
copies of $\ket{\pchoi(T')}$. To relate this set to $S_{\epsilon_0}(R/\eta)$, we have the bound $|r_a'| \leq |r_a| + \Delta(\varepsilon + \varepsilon_{\mathrm{amp}})$, which implies
\begin{equation}
\begin{split}
    |r_a| &\geq |r_a'| - \Delta(\varepsilon + \varepsilon_{\mathrm{amp}})\\
    &> \gamma - \Delta(\varepsilon + \varepsilon_{\mathrm{amp}})
\end{split}
\end{equation}
for all $a \in S_\gamma(2\Delta T')$. Thus if we choose $\gamma = \epsilon_0 + \Delta(\varepsilon + \varepsilon_{\mathrm{amp}})$, we get $S_{\epsilon_0}(R/\eta) \subseteq S_\gamma(2\Delta T')$.

At each step in the bootstrapping loop, we only need to consider terms in $H$ that we had not observed yet. Suppose the worst case, wherein at the current step we only have the minimally guaranteed set of terms $S_\eta(H)$. In other words, the estimate $\widehat{H}$ is supported only on $S_\eta(H)$, which means that
\begin{equation}
   R/\eta = \sum_{E_a \in S_{\eta}(H)} \l( \frac{\lambda_a - \widehat{\lambda}_a}{\eta} \r) E_a + \sum_{E_a \notin  S_{\eta}(H)} \frac{\lambda_a}{\eta} E_a.
\end{equation}
We focus on the terms in the second sum, which we have yet to identify. Upon learning $S_{\epsilon_0}(R/\eta)$, the new terms we have acquired are at least those which obey $|\lambda_a| > \epsilon_0 \eta = \eta/2$. Appending these newly identified terms with $S_\eta(H)$ gives us $S_{\eta/2}(H)$, as desired. Setting $\eta = 2^{-j}$ and recursing this procedure over $j = 0, 1, \ldots, \lfloor \log_2(1/\epsilon) \rfloor$ allows us to identify $S_\epsilon(H)$ with nearly Heisenberg-limited precision because the query complexity scales nearly linearly with $1/\eta$;~we show this in the sequel.

The argument is parallel to the analysis given in \cref{thm:structure_learning_base}. Take $\varepsilon = \varepsilon_{\mathrm{amp}} = c\epsilon_0/\Delta$ for a small constant $c$, so that $\gamma = \frac{1}{2} + 2c$ is also constant. The number of copies required is then $N = \Ot(\|r'\|^2)$. On the other hand, the success probability of preparing a copy from $B_{\mathrm{amp}}$ is $p = \frac{1}{2^n} \|T'\|_F^2 = \frac{1}{4} \|r'/\Delta\|^2$. Thus by \cref{prop:prob_state_prep}, $\O(N/p) = \Ot(\Delta^2) = \Ot(m^2)$ queries to $B_{\mathrm{amp}}$ suffice. Finally, recall from \cref{thm:amp_BE_residual} that we make $\Ot(1/\eta)$ queries to $\ctrl{}{e^{\pm \i H/\Delta}}$ in order to produce the amplified block encoding. In sum, we get the following result about structure learning from residuals.

\begin{theorem}\label{thm:amp_residual_structure}
    Let $H = \sum_{a=1}^m \lambda_a E_a$ be an unknown $n$-qubit Hamiltonian and $\widehat{H} = \sum_{a=1}^m \widehat{\lambda}_a E_a$ an estimate such that $\|\lambda - \widehat{\lambda}\|_\infty \leq \eta$ for some $\eta \in (0, 1]$. There exists a quantum algorithm which, with high probability, identifies all $E_a$ such that $|\lambda_a| \geq \eta/2$. The algorithm makes $\Ot(m^2/\eta)$ queries to \emph{$\ctrl{}{e^{\pm\i H/(2m)}}$} and uses $\poly(n, m, 1/\eta)$ quantum and classical computation.
\end{theorem}

\subsection{Heisenberg-limited bootstrap analysis}

Here, we finish the complexity analysis by placing these results into the outer loop described in \cref{alg:DOS_bootstrap}. For completeness, let us state and quickly prove why the bootstrap framework achieves Heisenberg-limited scaling.

\begin{proposition}\label{prop:bootstrap}
    Suppose that the base algorithm $\alg(H, \eta, \zeta)$ uses $\O(\tau \log(1/\zeta) / \eta^{\alpha})$ evolution time under $H$, for some $\tau > 0$ and $\alpha > 1$. Furthermore, suppose that $\alg$ has the property that $\alg(sH + H', \eta, \zeta)$ uses evolution time $\O(s\tau \log(1/\zeta) / \eta^{\alpha})$, for any $s > 0$ and any (known) Hermitian operator $H'$. Then, Algorithm~\ref{alg:DOS_bootstrap} returns an $\epsilon$-accurate estimate of the Hamiltonian coefficients with probability at least $1 - \delta$, using total evolution time $\O(\tau \log(1/\delta)/\epsilon)$.
\end{proposition}

\begin{proof}
First we analyze the query complexity of the algorithm. At each step $j$ in the loop, the base algorithm only requires a fixed error of $\epsilon_0 = 1/2$, so the dependence on precision $\O(1/\epsilon_0^\alpha)$ is reduced to a constant. Meanwhile, we scale $H$ by $s = 1/\eta_j$. Thus we query the black box for time $|t_j| = \O((\tau / \eta_j) \log(1/\zeta_j))$. With the choices $\eta_j = 2^{-j}$, $\zeta_j = \delta/2^{T+1-j}$, and $T = \lfloor \log_2(1/\epsilon) \rfloor$, the total evolution time is
\begin{equation}
\begin{split}
    \ttotal &= \sum_{j=0}^T |t_j| = \O\l( \tau \sum_{j=0}^T 2^j \log\l( \frac{2^{T+1-j}}{\delta} \r) \r)\\
    &= \O\l(\tau \sum_{j=0}^T 2^j ( T + 1 - j + \log(1/\delta) ) \r)\\
    &= \O(\tau 2^T \log(1/\delta)) = \O((\tau / \epsilon) \log(1/\delta)).
\end{split}
\end{equation}

Next we check the correctness of the algorithm. By a union bound, the overall failure probability is at most
\begin{equation}
    \sum_{j=0}^T \zeta_j = \delta \sum_{j=0}^T 2^{-(T+1)+j} \leq \delta.
\end{equation}
Conditioned on the success at each step, each estimate of the residual $\widehat{\lambda}^{(j)}$ is within $1/2$ of $(\lambda - \lambda^{(j)}) / \eta_j$. Equivalently, this means that
\begin{equation}
    \| \lambda - (\widehat{\lambda}^{(j)} + \eta_j \widehat{r}^{(j)}) \|_\infty \leq \frac{\eta_j}{2} = \frac{1}{2^{j+1}}.
\end{equation}
Thus at the final step, we are left with an estimate with error at most $2^{-(T+1)} \leq \epsilon$.
\end{proof}

Now we piece together the claim of \cref{thm:tr_learning}, up to the time resolution $\tmin$ which we address in \cref{sec:controlization}.

\begin{theorem}
    Let $H = \sum_{a=1}^m \lambda_a E_a$ be an unknown $n$-qubit Hamiltonian but $m$ given to us. Let $\epsilon, \delta \in (0, 1)$. There exists a structure learning algorithm that outputs $\widehat{\lambda} \in \R^m$, along with classical representations of corresponding $E_1, \ldots, E_m$, such that $\Pr(\|\lambda - \widehat{\lambda}\|_\infty \leq \epsilon) \geq 1 - \delta$. The algorithm uses a total evolution time of $\ttotal = \Ot(m\log(1/\epsilon)/\epsilon)$ to time-forward and time-reverse queries, and uses $\poly(n, m, 1/\epsilon, \log(1/\delta))$ quantum and classical computation.
\end{theorem}

\begin{proof}
    We combine \cref{thm:amp_residual_learning,thm:amp_residual_structure} to construct the base algorithm $\alg$. (For the $j = 0$ step, we can use \cref{prop:CW_learning_algorithm,thm:structure_learning_base} without using uniform spectral amplification, at the cost of slightly worse constant factors.) Let $R_j = H - \widehat{H}_j$, $\eta_j = 2^{-j}$, and $\zeta_j = \delta/2^{T+1-j}$. We first identify the terms in $R_j/\eta_j$ which are at least $1/2$ in magnitude, $\widehat{S}_j \supseteq S_{1/2}(R_j/\eta_j)$. By \cref{thm:amp_residual_structure}, this costs $\Ot(m^2 \log(1/\eta_j)/\eta_j)$ queries evolving for time $\frac{1}{2m}$. Then we estimate all $\widehat{r}_j$ corresponding to the Pauli terms in $\widehat{S}_j$, which by \cref{thm:amp_residual_learning} has the same query complexity. Repeating these procedures $\O(\log(1/\zeta_j))$ times and taking the median for each estimate suppresses the failure probability to at most $\zeta_j$.
    
    Thus, $\alg(R_j/\eta_j, 1/2, \zeta_j)$ uses evolution time
    \begin{equation}
        |t_j| = \Ot((m\log(1/\eta_j)/\eta_j) \log(1/\zeta_j)) \leq \Ot((m\log(1/\epsilon)/\eta_j) \log(1/\zeta_j)),
    \end{equation}
    where the $\Ot(\cdot)$ notation here suppresses logarithmic factors not involving $\eta_j$, $\zeta_j$, $\epsilon$, or $\delta$. Inserting this into the result of \cref{prop:bootstrap} with $\tau = m \log(1/\epsilon)$ yields the claim.
\end{proof}

%% file: time-reversal-free_algorithm.tex
\section{Structure learning without time reversal}\label{sec:time_reversal_free}

To prepare Choi-like states of $H$, we required a block encoding of $H/\Delta$ for some appropriate normalization $\Delta$. The approach described up until now relied heavily on the ability to perform both the time evolution operator and its inverse, in order to make full use of the QSVT toolbox. (We also relied on multi-controlled time evolutions, however the universal controlization scheme described in \cref{sec:controlization} allows us to circumvent that concern.) In the access model for this section, we no longer assume access to time reversal. While some time-reversal techniques exist, they either require knowledge of the Hamiltonian structure~\cite{odake2024higher} or scale exponentially in $n$~\cite{quintino2019reversing}. Absent that information at the start of the learning algorithm, it is unclear how to utilize those techniques in our setting. Even more, it has been shown that there exists unitary learning tasks for which time reversal enables exponential speedups in query complexity~\cite{schuster2023learning,cotler2023information,schuster2024random}. Such results establish that, in general, time reversal is a powerful resource in learning.

The goal of this section is to prove \cref{thm:tf_learning}, up to the time resolution $\tmin$ which will be addressed in \cref{sec:controlization}.

\subsection{Block encoding the matrix logarithm}

In this section we develop an alternative avenue for block encoding the logarithm of $e^{-\i H/\Delta}$, using only forward-time queries. For exposition, we assume throughout that we know the value of $\|H\|$ and so can choose $\Delta = 2\|H\|$;~without that knowledge, all appearances of the norm can be replaced by $m \geq \|H\|$. Our construction is based on the well-known power series for the matrix logarithm.

\begin{proposition}[\cite{hall2015lie}]
    For any $d \times d$ complex matrix $A$, define
    \begin{equation}
        \log A \coloneqq \sum_{k=1}^\infty (-1)^{k+1} \frac{(A - \I)^k}{k}.
    \end{equation}
    This series converges whenever $\|A - \I\| < 1$ and satisfies
    \begin{equation}
        e^{\log A} = A.
    \end{equation}
    Additionally, let $X \in \C^{d \times d}$ obeying $\|X\| < \log 2$. Then $\|e^X - \I\| < 1$ and
    \begin{equation}
        \log e^X = X.
    \end{equation}
\end{proposition}
For an integer $K \geq 1$, we define the $K$th partial sum of the series,
\begin{equation}\label{eq:matrix_log_partial_sum}
    L_K(A) \coloneqq \sum_{k=1}^K (-1)^{k+1} \frac{(A - \I)^k}{k}.
\end{equation}
The error of this truncation is easily bounded as follows.
\begin{lemma}\label{lem:matrix_log_error}
    Let $A \in \C^{d \times d}$ such that $\|A - \I\| \leq r < 1$. Then
    \begin{equation}
        \|{\log A} - L_K(A)\| \leq \frac{r^{K+1}}{(K+1)(1-r)}.
    \end{equation}
\end{lemma}


\begin{proof}
    Compute:
    \begin{equation}
    \begin{split}
        \|{\log A} - L_K(A)\| &= \l\| \sum_{k=K+1}^\infty \frac{(A - \I)^k}{k} \r\|\\
        &\leq \sum_{k=K+1}^\infty \frac{r^k}{k}\\
        &\leq \frac{r^{K+1}}{K+1} \sum_{k=0}^\infty r^k\\
        &= \frac{r^{K+1}}{(K+1)(1-r)}.
    \end{split}
    \end{equation}
\end{proof}


It remains to block encode $\i L_K(e^{-\i H/\Delta})$, where $\Delta$ is such that $\|H/\Delta\| < \log 2$. It turns out we will need a slightly stronger subnormalization of $\|H/\Delta\| \leq \frac{1}{2}$, as was the case with the $\arcsin$ polynomial approximation. Let us fix a $\Delta \geq 2\|H\|$ as usual. To simplify notation, we shall denote $U \equiv U(1/\Delta)$ throughout this section.

The follow lemma expresses the truncated power series for $\log U$ as a linear combination of powers of $U$.

\begin{lemma}\label{lem:matrix_log_LCU}
    The operator $L_K(U)$ can be represented as
    \begin{equation}
        L_K(U) = \sum_{j=0}^K c_j U^j,
    \end{equation}
    where the coefficients $c_j \in \R$ are given by
    \begin{equation}\label{eq:matrix_log_coeff}
    \begin{split}
        c_0 &= -\sum_{k=1}^K \frac{1}{k},\\
        c_j &= \frac{(-1)^{j+1}}{j} \binom{K}{j}, \quad j = 1, \ldots, K.
    \end{split}
    \end{equation}
    The $\ell_1$-norm of the coefficients, $\Lambda \coloneqq \sum_{j=0}^K |c_j|$, obeys the following bounds:
    \begin{equation}
        \Omega(2^K/K) \leq \Lambda \leq \O(2^K).
    \end{equation}
\end{lemma}

\begin{proof}
    Since all terms in \cref{eq:matrix_log_partial_sum} are just powers of $U$, we can treat $L_K(U)$ as a commutative polynomial of degree $K$. By binomial expansion of $(U - \I_n)^k$, we get
    \begin{equation}
    \begin{split}
        L_K(U) &= \sum_{k=1}^K \frac{(-1)^{k+1}}{k} (U - \I_n)^k\\
        &= \sum_{k=1}^K \frac{(-1)^{k+1}}{k} \sum_{j=0}^k \binom{k}{j} (-1)^{k-j} U^j\\
        &= \sum_{j=0}^K \sum_{k=1}^K \frac{(-1)^{j+1}}{k} \binom{k}{j} U^j,
    \end{split}
    \end{equation}
    where the rearrangement in the final line is valid because $\binom{k}{j} = 0$ whenever $k < j$. Thus the coefficients of the polynomial are
    \begin{equation}
        c_j = (-1)^{j+1} \sum_{k=1}^K \frac{1}{k} \binom{k}{j}.
    \end{equation}
    The value at $j = 0$ is clear. For $j \geq 1$, first we rewrite each summand as:
    \begin{equation}
    \begin{split}
        \frac{1}{k} \binom{k}{j} &= \frac{(k - 1)!}{j!(k - j)!}\\
        &= \frac{1}{j} \frac{(k - 1)!}{(j - 1)!(k - j)!}\\
        &= \frac{1}{j} \binom{k - 1}{j - 1}.
    \end{split}
    \end{equation}
    To evaluate the sum over $k$, we use the hockey-stick identity, which states that
    \begin{equation}
        \sum_{k'=j'}^{K'} \binom{k'}{j'} = \binom{K' + 1}{j' + 1}.
    \end{equation}
    Inserting the appropriate values of $j', k', K'$ for our context, we get
    \begin{equation}
       \sum_{k=1}^K \binom{k - 1}{j - 1} = \sum_{k=0}^{K-1} \binom{k}{j - 1} = \sum_{k=j-1}^{K-1} \binom{k}{j - 1} = \binom{K}{j},
    \end{equation}
    from which \cref{eq:matrix_log_coeff} follows. To get bounds on the scaling of $\Lambda$, the upper bound can be seen by
    \begin{equation}\label{eq:l1_norm_bound}
    \begin{split}
       \Lambda &= \sum_{j=0}^K |c_j| = |c_0| + \sum_{j=1}^K \frac{1}{j} \binom{K}{j}\\
       &\leq \log K + 1 + \sum_{j=1}^K \binom{K}{j}\\
       &= \log K + 2^K.
    \end{split}
    \end{equation}
    Meanwhile for the lower bound, we have
    \begin{equation}
    \begin{split}
        \Lambda &\geq |c_0| + \sum_{j=1}^K \frac{1}{K} \binom{K}{j}\\
        &\geq \log(K + 1) + \frac{1}{K}(2^K - 1),
    \end{split}
    \end{equation}
    which finishes the claim.
\end{proof}

\begin{remark}
    Ideally, we would prefer to have constructed the LCU for $\log U$ using block encodings of $(U - \I_n)^k$, wherein the coefficients satisfy $\sum_{k=1}^K |(-1)^{k+1}/k| = \O(\log K)$. However, each such block encoding is subnormalized by a factor of $2^k$, which would then have to be boosted, say, using singular value amplification~\cite{gilyen2019quantum} which again requires $U^\dagger$.
\end{remark}

From this result, we immediately get an $(\O(1/\gamma), \O(\log\log(1/\gamma), \O(\gamma))$-BE of $H/\Delta$ for some $\gamma > 0$ of our choosing, provided that $\|H/\Delta\| \leq \frac{1}{2}$.

\begin{theorem}\label{thm:matrix_log_H}
    Let $H$ be an $n$-qubit Hamiltonian and $\Delta \geq 2\|H\|$. Fix some integer $K \geq 1$ and denote $U = e^{-\i H/\Delta}$. From access to \emph{$\ctrl{}{U^j}$} for $j = 0, 1, \ldots, K$, we can construct a $(\Lambda, \lceil \log_2(K + 1) \rceil, 2^{-(K + 1)})$-BE of $H/\Delta$, where $\Lambda = \O(2^K)$. This construction costs $\O(K^2 / \Delta)$ evolution time to the Hamiltonian.
\end{theorem}

\begin{proof}
    \cref{lem:matrix_log_LCU} describes the LCU representation for $\i L_K(U)$ involving only $U^j$ for $j = 0, 1, \ldots, K$, with $\ell_1$-norm $\Lambda \leq \log K + 2^K$. \cref{prop:LCBE} gives the prescription for realizing the LCU circuit, using $\lceil \log_2(K + 1) \rceil$ control qubits. Finally, observe that if $\|H/\Delta\| \leq \frac{1}{2}$, then
    \begin{equation}
        \|e^{-\i H/\Delta} - \I_n\| \leq |e^{-\i/2} - 1| < \frac{1}{2}.
    \end{equation}
    Therefore if we take $r = 1/2$ in \cref{lem:matrix_log_error}, we get the error bound
    \begin{equation}
        \|H/\Delta - \i L_K(U)\| \leq \frac{2^{-K}}{K + 1} \leq 2^{-(K + 1)}.
    \end{equation}
    Note that the $j = 0$ term is merely the identity;~the evolution time from the remaining $K$ queries is $\sum_{j=1}^K j/\Delta = \O(K^2/\Delta)$.
\end{proof}

This block encoding of $H/\Delta$ comes with a small subnormalization of $\O(\gamma \log(1/\gamma))$. Unfortunately, without access to $U^\dagger$ it is difficult to perform spectral amplification to boost this up to $\Theta(1)$. Instead, we will resort to a more rudimentary approach, which is simply to prepare the (referenceless) pseudo-Choi states with low probability.

\subsection{Structure learning}

The following lemma bounds how much evolution time suffices to prepare $N$ copies of the referenceless pseudo-Choi state, $\ket{\pchoi(\i L_K(U))}$.

\begin{lemma}\label{lem:preparing_copies_of_pchoi}
    Let $H = \sum_{a=1}^m \lambda_a E_a$ be an $n$-qubit Hamiltonian, $\Delta \geq 2\|H\|$, and $U = e^{-\i H/\Delta}$. For any integers $N, K \geq 1$, with high probability we can prepare $N$ copies of $\ket{\pchoi(\i L_K(U))}$ from
    \begin{equation}
        Q = \O\l( \frac{4^K K N}{\frac{1}{2^n} \|\i L_K(U)\|_F^2} \r)
    \end{equation}
    queries and $\O(KQ/\Delta)$ evolution time to \emph{$\ctrl{}{U}$} (and no queries to its inverse).
\end{lemma}

\begin{proof}
    Let
    \begin{equation}
        B = \begin{pmatrix}
        \i L_K(U)/\Lambda & \cdot\ \\
        \cdot & \cdot\ 
    \end{pmatrix}
    \end{equation}
    be the block encoding described in \cref{thm:matrix_log_H}. Recall that each instance of this unitary requires $\O(K^2/\Delta)$ evolution time to construct. As is standard, we apply $B \otimes \I_n$ to the state $\ket{0^q} \ket{\Omega}$ with $q = \lceil \log_2(K + 1) \rceil$ and measure the $q$-qubit register. If we observe $\ket{0^q}$, then the remaining $2n$-qubit register contains $\ket{\pchoi(\i L_K(U))}$ as desired. This occurs with probability
    \begin{equation}
    \begin{split}
        \Pr(0^q) &= \l\|\l( \frac{\i L_K(U)}{\Lambda} \otimes \I_n \r) \ket{\Omega}\r{\|^2}\\
        &= \frac{1}{2^n \Lambda^2} \|\i L_K(U)\|^2_F.
    \end{split}
    \end{equation}
    By \cref{prop:prob_state_prep}, we need $\O(N/\Pr(0^q))$ queries to $B$ to prepare $N$ copies with high probability. Each query to $B$ costs $K + 1$ queries and $\O(K^2/\Delta)$ evolution time to $U$. Use the bound $\Lambda^2 \leq \O(4^K)$ from \cref{eq:l1_norm_bound} to conclude the statement.
\end{proof}

With this result, we are equipped to determine the complexity of structure learning arbitrary Hamiltonians, using our algorithm with only forward-time queries.

\begin{theorem}\label{thm:positive-time_SL}
    Fix $\epsilon > 0$. Let $S_\epsilon(H) = \{E_a : |\lambda_a| > \epsilon\}$ be the set of terms in $H = \sum_{a=1}^m \lambda_a E_a$ which have weight greater than $\epsilon$. Let $\Delta \geq 2\|H\|$. From only positive-time queries to $e^{-\i H t}$, $t \geq 0$, we can learn all the elements of $S_\epsilon$ using
    \begin{equation}
        \Ot\l( \frac{\Delta^4}{\epsilon^4} \r) \text{ queries and } \Ot\l( \frac{\Delta^3}{\epsilon^4} \r) \text{ evolution time,}
    \end{equation}
    with high probability.
\end{theorem}

\begin{proof}
    For notation, write $H' = \i\Delta L_K(U)$ which has the Pauli decomposition $H' = \sum_{a=1}^{m'} \lambda'_a E_a$ for some $m' \geq m$. We remark that $\lambda'_a \in \C$ because $\i L_K(U)$ is not guaranteed to be Hermitian, although this will not affect our analysis. Let $\epsilon' > 0$ to be determined later and define the set
    \begin{equation}
        S_{\epsilon'}(H') = \{ E_a : |\lambda'_a| > \epsilon' \}.
    \end{equation}
    We apply \cref{lem:coupon_collector} to the state $\ket{\pchoi(H')} = \ket{\pchoi(\i L_K(U))}$, allowing us to learn all the elements of $S_{\epsilon'}(H')$ with high probability. This requires
    \begin{equation}
        N = \O\l(\frac{\frac{1}{2^n} \|\i\Delta L_K(U)\|_F^2 \log m}{\epsilon'^2}\r)
    \end{equation}
    copies of the referenceless pseudo-Choi state. Combined with \cref{lem:preparing_copies_of_pchoi}, this costs $Q = \Ot(4^K K \Delta^2/\epsilon'^2)$ queries and $\Ot(KQ/\Delta)$ evolution time.

    Next, we need to determine how close $H'$ should be to $H$ to guarantee $S_\epsilon(H) \subseteq S_{\epsilon'}(H')$. Recall that if $\|H - H'\| \leq \gamma \Delta$, then the error in each coefficient is also bounded by $|\lambda_a - \lambda_a'| \leq \gamma \Delta$. Then as long as $|\lambda_a'| > \epsilon - \gamma \Delta$, we are ensured that $|\lambda_a| > \epsilon$. Set $\gamma\Delta = \epsilon/2$, which implies the choices $\epsilon' = \epsilon/2$ and
    \begin{equation}\label{eq:K_bound}
    \begin{split}
        K &= \l\lceil \log_2\l( \frac{1}{2\gamma} \r) \r\rceil\\
        &= \lceil \log_2( \Delta/\epsilon ) \rceil.
    \end{split}
    \end{equation}
    Thus $4^K K = \Ot(\Delta^2/\epsilon^2)$ and so the claim follows.
\end{proof}

\subsection{Parameter estimation}

Once we have identified the set $S_\epsilon(H)$, any learning algorithm appropriate for arbitrary Pauli terms can be deployed to learn the parameters to precision $\epsilon$. The most straightforward algorithm to use is that of \cite{caro2024learning}, which uses shadow tomography on copies of the proper Choi state $(U(t) \otimes \I_n) \ket{\Omega}$. After a total evolution time of $\Ot(\|H\|^3/\epsilon^4)$, the algorithm uses a polynomial interpolation postprocessing step to output $\epsilon$-accurate estimates to $\lambda_a$ for each $E_a \in S_\epsilon(H)$.

However, the shadow tomography step in \cite{caro2024learning} requires a large quantum memory of $\Ot(\|H\|^2/\epsilon^2)$ qubits. This quantum memory is used to perform a gentle measurement on multiple coherently stored copies of the Choi states. The size of this register can be reduced to $2n$ by recent advances in shadow tomography~\cite{king2024triply}, namely the concept of a so-called mimicking state. This mimicking state serves the same effective purpose as the many coherent copies, but would only be the size of the system it is mimicking. However, constructing such a mimicking state takes $2^{\O(n)}$ quantum (and classical) gates, so this approach would not be computationally efficient.

Instead, let us apply the pseudo-Choi learning algorithm of \cite{castaneda2023hamiltonian} directly to our construction. The preparation of the pseudo-Choi state with reference, from a query to $\ctrl{0}{B}$, still succeeds with probability at least $\frac{1}{2}$. However, the state we produce has small norm in the encoded Hamiltonian:
\begin{equation}
    \l| \pchoiref\l(\frac{\i L_K(U)}{\Lambda}\r) \r\rangle = \frac{(\frac{\i L_K(U)}{\Lambda} \otimes \I_n) \ket{\Omega} \ket{0} + \ket{\Omega} \ket{1}}{\norm(\frac{\i L_K(U)}{\Lambda})}.
\end{equation}
In this case, analogous to \cref{eq:pchoi_shadow_estimator} we define the estimator from classical shadows as
\begin{equation}
    \widehat{\lambda}_a \coloneqq \Lambda \Delta \frac{\Re[\widehat{o}_a]}{\Re[\widehat{o}_\norm]}.
\end{equation}
Then \cref{eq:CW_sample_complexity} gets the modification $\Delta \leftarrow \Lambda\Delta$, leaving us with a new shadow tomography copy complexity of
\begin{equation}\label{eq:copy_complexity_Lambda2}
    \Ot\l( \frac{\Lambda^2 \Delta^2}{\epsilon^2} \r) = \Ot\l( \frac{\Delta^2}{\gamma^2 \epsilon^2} \r).
\end{equation}
This learns $\lambda_a'$ to error $\epsilon$. Since $\lambda_a'$ is $(\gamma\Delta)$-close to $\lambda_a$, as usual we take $\gamma = c\epsilon/\Delta$ for sufficiently small $c > 0$ and improve our sampling error to $(1 - c)\epsilon/2$ as well. This gets us a query complexity
\begin{equation}
    Q = \Ot\l( \frac{\Delta^4}{\epsilon^4} \r)
\end{equation}
queries to $\ctrl{}{e^{-\i (H/\Delta) j}}$, for $0 \leq j \leq \O(\log(\Delta/\epsilon))$. Combined with structure identification which has the same costs (\cref{thm:positive-time_SL}), we get \cref{thm:tf_learning}, up the time resolution claim which we establish in \cref{sec:controlization}.

\begin{theorem}
    Let $H = \sum_{a=1}^m \lambda_a E_a$ be an unknown $n$-qubit Hamiltonian. For any $\Delta \geq 2\|H\|$ and $\epsilon \in (0, 1)$, there exists a structure learning algorithm which, with high probability, learns all $\lambda_a$ to additive error $\epsilon$. The algorithm requires no prior knowledge about the identity of the terms $E_a$, makes $\Ot(\Delta^4/\epsilon^4)$ queries to $e^{-\i H/\Delta}$ (and no queries to any form of $e^{+\i H/\Delta}$), and uses $\poly(n, m, 1/\epsilon)$ quantum and classical computation.
\end{theorem}

%% file: controlization.tex
\section{Approximate controlization of time evolutions}\label{sec:controlization}

In this section we describe how to approximately implement multi-controlled versions of $U(t) = e^{-\i H t}$, given only the ability to apply it for arbitrary $t$. While there are no-go results for this task in general~\cite{araujo2014quantum,gavorova2024topological}, our fractional-query access model (since we can tune $t$) makes this task possible. A number of approaches have been developed, suitable for different scenarios~\cite{janzing2002quantum,zhou2011adding,friis2014implementing,nakayama2015quantum,dong2019controlled,chowdhury2024controlization}. We take the one described in \cite{odake2023universal}, which is universally applicable to any time evolution operator. We also describe how to generalize their technique to multi-controlled operations, at no additional query cost.

For any $k$-bit string $b$ and $n$-qubit operator $O$, define
\begin{equation}
    \ctrl{b}{O} \coloneqq \op{b}{b} \otimes O + (\I_k - \op{b}{b}) \otimes \I_n.
\end{equation}
Our goal is to approximately implement $\ctrl{b}{e^{-\i H_0 t}}$ for any $b \in \{0, 1\}^k$, where $H_0$ is the traceless part of $H$. This is sufficient for our purposes, since we do not care about learning the identity component of the Hamiltonian. First, we will review the $k = 1$ case. Afterwards, we build up the generalization to $k \geq 1$ control qubits. Finally, we show how our learning algorithms should be modified to account for controlization.

\subsection{A single control qubit}

The main idea of controlization from \cite{odake2023universal} relies on the fact that twirling by the Pauli operators $\Pauli{n} = \{\I, X, Y, Z\}^{\otimes n}$ yields
\begin{equation}
    \frac{1}{4^n} \sum_{P \in \Pauli{n}} PHP = \frac{\tr H}{2^n} \I_n.
\end{equation}
For $b \in \{0, 1\}$, let $\overline{b} = b \oplus 1$. Appending a control qubit and replacing each Pauli gate with its controlled form (conditioned on $\overline{b}$), we get
\begin{equation}
\begin{split}
    \frac{1}{4^n} \sum_{P \in \Pauli{n}} \ctrl{\overline{b}}{P} (\I_1 \otimes H) \ctrl{\overline{b}}{P} &= \frac{1}{4^n} \sum_{P \in \Pauli{n}} \begin{pmatrix}
        P & 0\\
        0 & \I_n
    \end{pmatrix}
    \begin{pmatrix}
        H & 0\\
        0 & H
    \end{pmatrix}
    \begin{pmatrix}
        P & 0\\
        0 & \I_n
    \end{pmatrix}\\
    &= \begin{pmatrix}
        \tr(H)\I_n/2^n & 0\\
        0 & H
    \end{pmatrix}\\
    &= \begin{pmatrix}
        0 & 0\\
        0 & H_0
    \end{pmatrix} + \frac{\tr H}{2^n} \I_{n + 1},
\end{split}
\end{equation}
where $H_0 = H - \tr(H)\I_n/2^n$ is the traceless part of $H$. Denote this $(n + 1)$-qubit operator by $H^{\Pauli{n},1}$. It generates the unitary
\begin{equation}
    \exp(-\i H^{\Pauli{n},1} t) = e^{\i \alpha} \begin{pmatrix}
        \I_n & 0\\
        0 & e^{-\i H_0 t}
    \end{pmatrix} = e^{\i \alpha} \ctrl{b}{e^{-\i H_0 t}},
\end{equation}
where $\alpha = -t \tr(H)/2^n$. Thus up to this global phase, we have a formalism for the desired controlled unitary. This can realized by a product formula over the $4^n$ terms of $H^{\Pauli{n},1}$, each of which only requires the uncontrolled time evolution:
\begin{equation}\label{eq:ctrl_product_formula_term}
    \exp\l(-\frac{\i t}{4^n} \, \ctrl{\overline{b}}{P} (\I_1 \otimes H) \ctrl{\overline{b}}{P} \r) = \ctrl{\overline{b}}{P} (\I_1 \otimes e^{-\i H t / 4^n}) \ctrl{\overline{b}}{P}.
\end{equation}
Taking the product over all $P \in \Pauli{n}$ yields a product-formula approximation to $\ctrl{b}{e^{-\i H_0 t}}$. However, this naive implementation clearly requires an exponential amount of resources, as well as an exponentially small time resolution.

To circumvent this issue, one can instead appeal to the randomized product formula qDRIFT~\cite{campbell2019random}. In this approach the gate sequence is constructed by randomly sampling terms, whereby the distribution of terms is proportional to their operator norms. In this case, the distribution reduces to uniformly drawing elements from $\Pauli{n}$. Consider a qDRIFT sequence constructed from $N$ such random draws. That is, we have the circuit $V_1 \cdots V_N$ where each $V_i$ is of the form
\begin{equation}
    V = \ctrl{\overline{b}}{P} (\I_1 \otimes U(t/N)) \ctrl{\overline{b}}{P}.
\end{equation}
This random construction can be viewed as the mixed-unitary channel $\E[\mathcal{V}_1 \circ \cdots \circ \mathcal{V}_N]$, where $\mathcal{V}_i(\cdot) = V_i (\cdot) V_i^\dagger$. Also, let $\mathcal{C}_b(t)$ be the quantum channel corresponding to the target unitary $\ctrl{b}{e^{-\i H_0 t}}$. The improved analysis of \cite{chen2021concentration} guarantees that the qDRIFT channel converges to the desired controlled unitary as
\begin{equation}
    \frac{1}{2} \| \mathcal{C}_b(t) - \E[\mathcal{V}_1 \circ \cdots \circ \mathcal{V}_N] \|_\diamond \leq \frac{t^2 \|H\|^2}{N},
\end{equation}
where $\| \cdot \|_\diamond$ is the diamond norm. Equivalently, building qDRIFT sequences of length $N = \O(t^2\|H\|^2/\gamma)$ suffices to achieve error $\gamma$.

\subsection{Multiple control qubits}

To generalize to multiple controls, first note that in the $k = 1$ setting, we have
\begin{equation}
    \ctrl{\overline{b}}{P} = \ctrl{b}{P} (\I_1 \otimes P).
\end{equation}
This follows because $P^2 = \I_n$. Now if we let $k \geq 1$, then for any $b \in \{0, 1\}^k$, the multi-controlled gates that we need to implement are
\begin{equation}
\begin{split}
    \ctrl{b}{P}(\I_k \otimes P) &= (\I_k - \op{b}{b}) \otimes P + \op{b}{b} \otimes \I_n\\
    &= \begin{pmatrix}
        P & \\
        & \ddots \\
        & & P\\
        & & & \I_n
    \end{pmatrix},
\end{split}
\end{equation}
where there are $2^k - 1$ blocks of $P$ above. Then, twirling over these gates gives us
\begin{equation}
\begin{split}
    H^{\Pauli{n},k} &= \frac{1}{4^n} \sum_{P \in \Pauli{n}} \ctrl{b}{P} (\I_k \otimes P) (\I_k \otimes H) (\I_k \otimes P) \ctrl{b}{P}\\
    &= \begin{pmatrix}
        0 & \\
        & \ddots \\
        & & 0\\
        & & & H_0
    \end{pmatrix} + \frac{\tr H}{2^n} \I_{n+k}.
\end{split}
\end{equation}
Simulating this twirled Hamiltonian with the qDRIFT protocol then yields an approximation to $\ctrl{b}{e^{-\i H_0 t}}$ with the same performance guarantees as in the single-control setting.

\begin{theorem}\label{thm:controlization}
    Let $H$ be a Hamiltonian and suppose we have access to $U(t) = e^{-\i H t}$. Let \emph{$\mathcal{C}_b(t)(\cdot) = \ctrl{b}{U(t)} (\cdot) \ctrl{b}{U(t)}^\dagger$} for any $b \in \{0,1\}^k$. For a $\gamma > 0$, there exists a mixed-unitary channel $\mathcal{E}(t)$ such that
    \begin{equation}
        \frac{1}{2} \|\mathcal{C}_b(t) - \mathcal{E}(t)\|_\diamond \leq \gamma,
    \end{equation}
    implemented from $N$ queries to $U(t/N)$ and $\O(N)$ Pauli and controlled-Pauli gates, where
    \begin{equation}
        N = \O\l(\frac{t^2 \|H\|^2}{\gamma}\r).
    \end{equation}
\end{theorem}

\begin{remark}[On the trace of the Hamiltonian]\label{rem:ctrl_traceless}
    This method of controlization isolates the traceless part of the Hamiltonian, so without further modifications it can only approximate $\ctrl{b}{e^{-\i H_0 t}}$. For our purposes this is actually favorable, because we can now always safely assume that the Hamiltonian to be learned is traceless, removing any potential complications with a nonzero identity component in $H$.
\end{remark}

\subsection{Error in pseudo-Choi state preparation}

All the algorithms presented in this work (even the time-reversal-free approach) require controlled variants of the time evolution operator. Using the controlization scheme of this section, all instances of $\ctrl{}{e^{-\i H t}}$ appearing in this paper can be approximated by the qDRIFT circuit involving $N$ applications of $e^{-\i H t/N}$. Let us briefly discuss how this approximation affects the learning guarantees.

To begin, observe that each query to $\ctrl{}{U(t)}$ runs for time $|t| = 1/\Delta$ when given access to time reversal, and $t = j/\Delta$ for $j \leq \O(\log(\Delta/\epsilon))$ without time reversal. Thus to streamline exposition, suppose $|t| = \Ot(1/\Delta)$. According to \cref{thm:controlization}, each query costs
\begin{equation}
    N = \Ot\l( \frac{\|H\|^2}{\Delta^2 \gamma} \r) = \Ot(1/\gamma)
\end{equation}
queries to the uncontrolled evolution $U(t/N)$, where the second bound follows from $\|H\| \leq 2\Delta$. Clearly, the total evolution time is the same. However, the time resolution is shortened to $t_\mathrm{min} = t/N = \Omega(\gamma/\Delta)$, so it behooves us to analyze what $\gamma > 0$ suffices. The controlization also introduces additional quantum gates, but this overhead is at most $\poly(n, \Delta, 1/\gamma)$ so we do not closely analyze it here.

\begin{theorem}\label{thm:controlization_time_resolution}
    In order to replace controlled time evolutions with approximate controlizations in any pseudo-Choi state learning algorithm, we can make $N$ queries to $U(t/N)$ per instance of \emph{$\ctrl{}{U(t)}$}. Letting $\epsilon > 0$ be the target $\ell_\infty$ error and $\Delta$ the choice of Hamiltonian normalization, it suffices to take for:
    \begin{itemize}
        \item Parameter estimation with time reversal:~$N = \Ot(\Delta)$;

        \item Parameter estimation without time reversal:~$N = \Ot(\Delta/\epsilon)$;

        \item Structure learning with time reversal:~$N = \Ot(\Delta^2/\epsilon)$;

        \item Structure learning without time reversal:~$N = \Ot(\Delta^4/\epsilon^4)$.
    \end{itemize}
    The time resolution becomes $\tmin = \Omega(1/(N\Delta))$, and the total evolution time and number of experiments are unaffected.
\end{theorem}

We will prove this statement in two parts:~first the effect of controlization on the parameter estimation stage, and then the effect on the term identification stage.

\begin{lemma}\label{lem:controlization_parameter_effect}
    Let $B$ be an $(\alpha, q, 0)$-BE that prepares any pseudo-Choi state with reference, and $\mathcal{B}$ its corresponding channel. Suppose $B$ makes $Q_1$ queries to controlled time evolutions. Let $\widetilde{\mathcal{B}}$ be the channel obtained by replacing all such queries with controlized approximations, each with diamond error $\gamma$. Let $\rho_0$ be any initial $(2n + 1 + q)$-qubit state, and define
    \begin{equation}
        \sigma \coloneqq \frac{(\bra{0^q} \otimes \I_{2n+1}) \mathcal{B}(\rho_0) (\ket{0^q} \otimes \I_{2n+1})}{\tr[(\op{0^q}{0^q} \otimes \I_{2n+1}) \mathcal{B}(\rho_0)]},
    \end{equation}
    Analogously, define $\widetilde{\sigma}$ with respect to $\widetilde{\mathcal{B}}$. Then for any decoding operator $O_j \in \{O_1,\ldots, O_m, O_\norm\}$, we have
    \begin{equation}
        |{\tr(O_j \sigma)} - \tr(O_j \widetilde{\sigma})| \leq 6 Q_1 \gamma.
    \end{equation}
\end{lemma}

\begin{proof}
    Let $\mathcal{C}_j$ denote a (multi-)controlled instance of $\mathcal{U}$ appearing in $\mathcal{B}$. Let $\mathcal{E}_j$ be the mixed-unitary channel guaranteed by \cref{thm:controlization} to approximate $\mathcal{C}_j$ to within diamond distance $\gamma$. By a telescoping inequality, we have
    \begin{equation}\label{eq:telescoped_ineq}
        \|\mathcal{B} - \widetilde{\mathcal{B}}\|_\diamond \leq \sum_{j=1}^{Q_1} \|\mathcal{C}_j - \mathcal{E}_j\|_\diamond \leq Q_1\gamma.
    \end{equation}
    For convenience, we will define $\rho \coloneqq \mathcal{B}(\rho_0)$ and $\widetilde{\rho} \coloneqq \widetilde{\mathcal{B}}(\rho_0)$ To incorporate the conditional state preparation into the observable estimation analysis, we use the fact that in expectation,
    \begin{equation}
       \tr(O_j \sigma) = \tr\l[ \l(\frac{\op{0^q}{0^q}}{\Pr(0^q; \rho)} \otimes O_j \r) \rho \r],
    \end{equation}
    where $\Pr(0^q; \rho) = \tr[(\op{0^q}{0^q} \otimes \I_{2n+1}) \rho]$. Analogously, we also have
    \begin{equation}
        \tr(O_j \widetilde{\sigma}) = \tr\l[ \l(\frac{\op{0^q}{0^q}}{\Pr(0^q; \widetilde{\rho})} \otimes O_j \r) \widetilde{\rho} \r].
    \end{equation}
    To simplify notation, let $a = \Pr(0^q; \rho)$ and $b = \Pr(0^q; \widetilde{\rho})$. By H\"{o}lder's inequality with $\|\op{0^q}{0^q} \otimes O_j\| \leq 1$, we get
    \begin{equation}
    \begin{split}
        |{\tr(O_j \sigma)} - \tr(O_j \widetilde{\sigma})| &= \tr\l[ \l( \op{0^q}{0^q} \otimes O_j \r) \l( \frac{\rho}{a} - \frac{\widetilde{\rho}}{b} \r) \r]\\
        &\leq \l\| \frac{\rho}{a} - \frac{\widetilde{\rho}}{b} \r{\|_1} = \frac{1}{a} \l\| \rho - \frac{a}{b} \widetilde{\rho} \r{\|_1}\\
        &\leq \frac{1}{a} \l( \|\rho - \widetilde{\rho}\|_1 + \l\| \widetilde{\rho} - \frac{a}{b} \widetilde{\rho} \r{\|_1} \r)\\
        &= \frac{1}{a} \|\rho - \widetilde{\rho}\|_1 + \l| \frac{1}{a} - \frac{1}{b} \r|.
    \end{split}
    \end{equation}
    Recall that preparing pseudo-Choi states with reference always succeeds with probability $a, b \geq \frac{1}{2}$. By definition of the diamond norm, the first term is bounded by
    \begin{equation}\label{eq:first_diamond_error_term}
    \begin{split}
        \frac{1}{a} \|\rho - \widetilde{\rho}\|_1 &\leq 2 \|\mathcal{B} - \widetilde{\mathcal{B}}\|_\diamond\\
        &\leq 2 Q_1\gamma.
    \end{split}
    \end{equation}
    For the second term, we apply another round of H\"{o}lder's inequality:
    \begin{equation}\label{eq:second_diamond_error_term}
    \begin{split}
        \l| \frac{1}{a} - \frac{1}{b} \r| &= \frac{|a - b|}{ab}\\
        &\leq 4 \l| {\tr[(\op{0^q}{0^q} \otimes \I_{2n+1}) (\rho - \widetilde{\rho})]} \r|\\
        &\leq 4 \|\rho - \widetilde{\rho}\|_1\\
        &\leq 4 Q_1\gamma.
    \end{split}
    \end{equation}
    Combining \cref{eq:first_diamond_error_term,eq:second_diamond_error_term} yields the desired bound.
\end{proof}

Next we show how controlization affects the magnitude of the encoded Pauli coefficients of $H$.

\begin{lemma}\label{lem:controlization_term_effect}
    Let $B$ be an $(\alpha, q, 0)$-BE that prepares the referenceless pseudo-Choi state of any $H$. Fix $\rho_0 = \ket{0^q}\op{\Omega}{\Omega}\bra{0^q}$, and let $Q_1, \mathcal{B}, \widetilde{\mathcal{B}}, \gamma$ be analogously as in \cref{lem:controlization_parameter_effect}. For any $P \in \Pauli{n}$, define $\lambda_P \coloneqq \frac{1}{2^n} \tr(PH)$ and $w_P \coloneqq \alpha \sqrt{\tr[(\op{0^q}{0^q} \otimes \op{P}{P}) \widetilde{\mathcal{B}}(\rho_0) ]}$. These quantities obey
    \begin{equation}
        |\lambda_P| \geq w_P - \alpha\sqrt{Q_1 \gamma}.
    \end{equation}
\end{lemma}

\begin{proof}
    We start by reviewing the properties of the target state $\sigma = \op{\pchoi(H)}{\pchoi(H)}$ from our analysis in \cref{sec:structure_learning}. Recall the basis states $\ket{P} \equiv (P \otimes \I_n) \ket{\Omega}$ for $P \in \Pauli{n}$, which define the distribution of Bell measurements:
    \begin{equation}
        \Pr(P; \sigma) = \ev{\sigma}{P} = \frac{|\lambda_P|^2}{\ell^2},
    \end{equation}
    where $\lambda_P = \frac{1}{2^n} \tr(PH)$ and $\ell = \|\lambda\|$. We also have the probability of preparing $\sigma$ from measuring the ancilla register of $\rho = \mathcal{B}(\rho_0)$,
    \begin{equation}
        \Pr(0^q; \rho) = \frac{\|(H \otimes \I_n) \ket{\Omega}\|^2}{\alpha^2} = \frac{\ell^2}{\alpha^2}.
    \end{equation}
    Multiplying these two quantities allows us to express $|\lambda_P|^2$ as
    \begin{equation}
    \begin{split}
        |\lambda_P|^2 &= \alpha^2 \Pr(P; \sigma) \Pr(0^q; \rho)\\
        &= \alpha^2 \tr\l[(\op{0^q}{0^q} \otimes \op{P}{P}) \rho \r].
    \end{split}
    \end{equation}
    This observation motivates our definition for $w_P = \alpha \sqrt{\tr[(\op{0^q}{0^q} \otimes \op{P}{P}) \widetilde{\mathcal{B}}(\rho_0) ]}$ as the magnitude of the Pauli coefficients encoded in the controlized form of the block encoding. The error of this approximation is bounded by
    \begin{equation}
        |w_P^2 - |\lambda_P|^2| \leq \alpha^2 \|\mathcal{B}(\rho_0) - \widetilde{\mathcal{B}}(\rho_0)\|_1 \leq \alpha^2 \|\mathcal{B} - \widetilde{\mathcal{B}}\|_\diamond \leq \alpha^2 Q_1 \gamma,
    \end{equation}
    where we applied the same telescoping argument from \cref{eq:telescoped_ineq} in the final inequality. Then we use $\sqrt{x + y} \leq \sqrt{x} + \sqrt{y}$ for $x, y \geq 0$ to get
    \begin{equation}
    \begin{split}
        w_P &\leq \sqrt{|\lambda_P|^2 + \alpha^2 Q_1 \gamma}\\
        &\leq |\lambda_P| + \alpha\sqrt{Q_1 \gamma},
    \end{split}
    \end{equation}
    which implies the claim.
\end{proof}

With these two lemmas in hand, we are ready to prove \cref{thm:controlization_time_resolution}.

\begin{proof}[Proof (of \cref{thm:controlization_time_resolution}).]
    \textbf{Parameter estimation.} Recall from \cref{eq:CW_sample_complexity} and the surrounding discussion that, in order to get $\epsilon$ precision in $\widehat{\lambda}_a$, we need $\Omega(\epsilon/\Delta)$ precision in the decoding operators $\langle O_j \rangle$. In the bootstrap framework (time-reversal), we only need a constant precision $\epsilon_0 = \frac{1}{2}$, so we can tolerate a systematic error of $|{\tr(O_j \sigma)} - \tr(O_j \widetilde{\sigma})| \leq \O(1/\Delta)$. Therefore by \cref{lem:controlization_parameter_effect}, we should choose $\gamma = \Theta(1/(Q_1\Delta))$ with a sufficiently small constant factor. Recall that $Q_1 = \Ot(1/\eta)$, dominated by the cost of spectral amplification on the residual Hamiltonian (\cref{thm:amp_BE_residual}). Since $\eta \geq \epsilon$, we get $t_{\mathrm{min}} = \Omega(\gamma/\Delta) = \Omegat(\epsilon/\Delta^2)$.
    
    Without the bootstrap framework (time-forward), our error tolerance is $\epsilon$-dependent, $|{\tr(O_j \sigma)} - \tr(O_j \widetilde{\sigma})| \leq \O(\epsilon/\Delta)$. However, since the number of queries in this setting is only $Q_1 = \O(\log(\Delta/\epsilon)) = \Ot(1)$, we end up with the same asymptotic time resolution, $t_{\mathrm{min}} = \Omegat(\epsilon/\Delta^2)$. Note that this applies to the originally proposed algorithm of \cite{castaneda2023hamiltonian}, since it also has a query complexity of $Q_1 = \O(\log(\Delta/\epsilon))$.

    \textbf{Structure learning.} The Bell sampling protocol collects Pauli terms which have large coefficients, and the result of \cref{lem:controlization_term_effect} tells us how to appropriately adjust the threshold for what we mean by ``large'' under controlization. In the bootstrap framework, recall that we learn the residual Hamiltonian, so make the notation change $\lambda_P \to r_P$ from \cref{lem:controlization_term_effect}. Also, from \cref{thm:amp_BE_residual} we have that $\alpha = 2\Delta$ and $Q_1 = \Ot(1/\eta)$. The algorithm only needs to find $P$ such that $|r_P| > \frac{1}{2}$, so we run the protocol with the threshold $w_P > \frac{1}{2} + \alpha\sqrt{Q_1 \gamma}$. This holds with $\gamma = \Theta(1/(\alpha^2 Q_1))$ with a sufficiently small constant, which implies $\tmin = \Omegat(\eta/\Delta^3) = \Omegat(\epsilon/m^3)$.

    Without the bootstrap framework, we have instead $\alpha = \Lambda\Delta = \O(\Delta^2/\epsilon)$ and a desired threshold $|\lambda_P| > \epsilon$. We can achieve this by taking $w_P > \epsilon + \alpha \sqrt{Q_1 \gamma}$, where $Q_1 = \O(\log(\epsilon/\Delta))$ (\cref{thm:matrix_log_H}). Hence it suffices to take $\gamma = \Theta(\epsilon^2/(\alpha^2 Q_1)) = \Thetat(\epsilon^4/\Delta^4)$, which results in a time resolution of $\tmin = \Omegat(\epsilon^4/\Delta^5)$.
\end{proof}

%% file: acknowledgments.tex
\section*{Acknowledgments}
\addcontentsline{toc}{section}{Acknowledgments}

AZ thanks Andrew Baczewski, Matthias Caro, John Kallaugher, Danial Motlagh, and Ojas Parekh for helpful discussions, as well as multiple anonymous reviewers for their thorough feedback. This work was supported by the Laboratory Directed Research and Development program at Sandia National Laboratories, under the Gil Herrera Fellowship in Quantum Information Science. Sandia National Laboratories is a multimission laboratory managed and operated by National Technology and Engineering Solutions of Sandia, LLC, a wholly owned subsidiary of Honeywell International, Inc., for the U.S.~Department of Energy’s National Nuclear Security Administration under contract DE-NA-0003525. AZ also acknowledges support from the U.S.~Department of Energy, Office of Science, Office of Advanced Scientific Computing Research, Accelerated Research in Quantum Computing.

%% file: appendix_chernoff.tex
\section{Copies to queries reduction}\label{sec:bernoulli_bound}

Let $p \in (0, 1]$ be the probability of preparing a state $\ket{\psi}$, say from querying a block encoding $B$ and measuring some of its ancilla qubits. As a series of Bernoulli trials, it is a standard result that $Q = \Theta(N/p)$ queries is necessary and sufficient to successfully prepare $N$ copies of the state. In this appendix we prove the statement for completeness, adapting the argument of \cite[Appendix B.6]{castaneda2023hamiltonian} to the appropriate level of generality for our purposes.

\begin{lemma}\label{lem:bernoulli_bound}
    Fix integers $Q \geq N \geq 1$.  Let $X_1, \ldots, X_Q$ be i.i.d.~Bernoulli random variables with $\Pr(X_i) = p \in (0, 1]$. Define $X = X_1 + \cdots + X_Q$. In order to observe the event $X \geq N$ with probability at least $1 - e^{-\Omega(N)}$,
    \begin{equation}
        Q = \Theta(N/p)
    \end{equation}
    trials are necessary and sufficient.
\end{lemma}

\begin{proof}
    We invoke Chernoff's inequality to bound the tail probability for the event $X < N$. For any $t \in (0, 1)$, the multiplicative Chernoff bound states that
    \begin{equation}
        \Pr(X \leq (1 - t) \E[X]) \leq \exp\l( -\frac{t^2 \E[X]}{2} \r).
    \end{equation}
    Since $\E[X] = Qp$, we can set $t = 1 - \frac{N - 1}{Qp}$ to get
    \begin{equation}
        \Pr(X \leq N - 1) \leq \exp\l( -\frac{Qp}{2} \l( 1 - \frac{N}{Qp} \r{)^2} \r).
    \end{equation}
    This bounds the failure probability of the entire process, so set the rhs to be equal to $e^{-cN}$ for some constant $c$. Taking logs and rearranging yields
    \begin{equation}
    \begin{split}
        \log(e^{cN}) &= \frac{Qp}{2} - N + \frac{N^2}{2Qp}\\
        &\geq \frac{Qp}{2} - N.
    \end{split}
    \end{equation}
    Solving for $Q$, we find
    \begin{equation}
        Q \leq \frac{2(c + 1)N}{p} = \O(N/p)
    \end{equation}
    is sufficient. Conversely, observe that for $t > 0$ to hold, we must have $Qp > N - 1$, so $Q = \Omega(N/p)$ is also necessary.
    
    Finally, we remark that $t < 1$ only when $N \geq 2$, so let us consider the $N = 1$ case separately. It is clear that $\Pr(X \leq 0) = \Pr(X = 0) = (1 - p)^Q$. If $p = 1$ then the result is trivial;~meanwhile for each $p \in (0, 1)$, there exists some $c' > 1$ such that $1 - p \geq e^{-c'p}$. Combined with $1 - p \leq e^{-p}$, we get
    \begin{equation}
        \Pr(X = 0) = e^{-\Theta(p) Q} = e^{-c},
    \end{equation}
    which implies that $Q = \Theta(1/p)$ when $N = 1$ as well. Note that even when $N = 1$ we can choose $c = \log(1/\delta)$ to get some small failure probability $\delta$.
\end{proof}

%% file: galactic_algorithm.tex
\section{The galactic algorithm}\label{sec:galactic_algorithm}

Here we prove \cref{cor:galactic_algorithm}. Let us restate it below.

\begin{corollary}[Galactic structure learning without time reversal]
    There exists a quantum algorithm that solves \cref{prob:1} under the discrete control access model alone. For any real constant $p \geq 1$, the algorithm queries $e^{-\i H t}$ for a total evolution time of $\ttotal = \Ot((2^p \|H\|)^{1+2/p}/\epsilon^{2+2/p})$, requires a time resolution of $t \geq \tmin = \Omegat(\epsilon^4/(2^p \|H\|)^5)$, runs $\Nexp = \Ot((2^p \|H\|)^{2+2/p}/\epsilon^{2+2/p})$ experiments, and uses $\poly(n, m, 1/\epsilon)$ quantum and classical computation. The discrete control operations act on an ancillary computational register of $\nanc = n + \O(\log\log(\|H\|/\epsilon))$ qubits.
\end{corollary}

\begin{proof}
    The key modification to make is shrinking the bound $r$ from \cref{lem:matrix_log_error}. Let $p \geq 1$ be some real constant and for simplicity suppose $\Delta = 2^p \|H\|$. Then $\|e^{-\i H/\Delta} - \I_n\| \leq 2^{-p} \equiv r$, which by \cref{lem:matrix_log_error} means that
    \begin{equation}
    \begin{split}
        \|H/\Delta - L_K(U)\| &\leq \frac{2^{-p(K + 1)}}{(K + 1)(1 - 2^{-p})}\\
        &\leq C \cdot 2^{-p(K + 1)}),
    \end{split}
    \end{equation}
    where $C^{-1} = 2(1 - 2^{-p})$ is a constant. To make this error at most $\gamma$, we can take $K = \lceil \log_2(C/\gamma)/p \rceil - 1$. Taking $\gamma = \Theta(\epsilon/\Delta)$, we get
    \begin{equation}
    \begin{split}
        \Lambda^2 = \O(4^K) &= \O((C/\gamma)^{2/p})\\
        &= \O((\Delta/\epsilon)^{2/p})\\
        &= \O((2^p \|H\|/\epsilon)^{2/p}).
    \end{split}
    \end{equation}
    Then recall from \cref{thm:positive-time_SL,eq:copy_complexity_Lambda2} that the query complexity for either parameter or structure learning is $\O(\Lambda^2 \Delta^2/\epsilon^2)$, and the time resolution from \cref{thm:controlization_time_resolution} is $\Omegat(\epsilon^4/\Delta^5)$.
\end{proof}

\begin{remark}
    One may be concerned that $K = 0$ for large $p$, but this is not an issue:~$K \geq 1$ always holds because $\lim_{p\to\infty}\lceil \log(2^p)/p \rceil = 2$. On the other hand, this means that for very large $p$, we are essentially just constructing a block encoding of $\i(U - \I_n) = H/\Delta + \O(4^{-p})$, which is why this approximation is so accurate for just two terms of the matrix logarithm. Of course, the trade-off is that $1/\Delta$ suppresses the estimation of the parameters of $H$ by a factor of $2^{-p}$ as well, hence the blow-up in this constant.
\end{remark}